\documentclass[11pt]{article}

\usepackage[margin=1in]{geometry}
\usepackage{amsmath,amssymb,amsthm,mathtools}
\usepackage{microtype}
\usepackage{enumitem}
\usepackage{xcolor}
\usepackage[colorlinks=true,linkcolor=blue!55!black,citecolor=blue!55!black,urlcolor=blue!55!black]{hyperref}
\usepackage[ruled,vlined]{algorithm2e}
\usepackage{tikz}
\usetikzlibrary{positioning,decorations.pathreplacing}

\SetKw{Return}{return}
\SetKw{NextTrial}{continue with the next trial}

\newtheorem*{theorem*}{Theorem}

\newtheorem{theorem}{Theorem}[section]
\newtheorem{lemma}[theorem]{Lemma}
\newtheorem{corollary}[theorem]{Corollary}
\newtheorem{proposition}[theorem]{Proposition}
\theoremstyle{definition}
\newtheorem{definition}[theorem]{Definition}

\newcommand{\Oh}{O^*}

\newcommand{\tw}{\operatorname{Tower}}

\title{$k$-Coloring is Faster than Computing the Chromatic Number}
\author{Or Zamir \\ Tel Aviv University}
\date{}

\begin{document}
\maketitle

\begin{abstract}
We prove that $k$-coloring on $n$-vertex graphs has a randomized algorithm running in time $ (2-\varepsilon_k)^n, $ where $\varepsilon_k>0$ for every fixed~$k$. Previously, only the cases~$k\leq 6$ were known to have faster solutions than the general~$\Oh\bigl(2^n\bigr)$ time algorithm of [Bj\"{o}rklund, Husfeldt, Koivisto, SICOMP 2009] that computes the chromatic number.
In fact, our algorithm solves the more general problem of $k$-list-coloring, even when the overall color palette is of unbounded size.

We resolve this long-standing open problem by generalizing and combining tools from the $(k+2)$-coloring to $k$-list-coloring reduction of [Zamir, ICALP 2021] and the hypergraph-containers based approach in [Zamir, STOC 2023]. Together with new algorithms for list-coloring instances mixing long and short color lists, this yields an iterable reduction from \((k+1)\)-list-coloring to \(k\)-list-coloring over arbitrary palettes. 
\end{abstract}

\newpage
\tableofcontents
\newpage

\section{Introduction}
The problem of $k$-coloring a graph, or determining its \emph{chromatic number} is one of the most fundamental and well-studied NP-complete problems. It was already listed as one of the first NP-complete problems in Karp's seminal 1972 paper~\cite{karp1972reducibility}. Similarly to $k$-SAT, the problem of $2$-coloring is polynomial, yet $k$-coloring is NP-complete for every $k\geq 3$~\cite{lovasz1973coverings,stockmeyer1973planar}.

There is substantial work exploring exponential-time worst-case algorithms for NP-Complete problems. A 2003 survey of Woeginger \cite{woeginger2003exact} covers and refers to dozens of papers exploring such algorithms for many problems including satisfiability, graph coloring, knapsack, TSP, maximum independent sets and more. A subsequent survey of Fomin and Kaski \cite{fomin2013exact} and a book of Fomin and Kratsch \cite{fomin2010exact} further cover the topic. More recently, the study of these exact running times has become closely connected with \emph{fine-grained complexity}. Rather than distinguishing only between polynomial and superpolynomial running times, fine-grained complexity supplements qualitative assumptions such as \(\mathrm{P}\neq\mathrm{NP}\) with quantitative hypotheses asserting that particular benchmark problems cannot be solved substantially faster than their best known algorithms~\cite{ImPa01,impagliazzoPaturiZane2001,cyganEtAl2016,vassilevskaWilliamsWilliams2018,abboudVassilevskaWilliams2014,backursIndyk2015,bringmann2014frechet,vassilevskaWilliams2018}.

For SAT, the straightforward enumeration algorithm runs in $O^*(2^n)$ time. On the other hand, it is known that for every fixed $k$ there exists a constant $\varepsilon_k>0$ such that $k$-SAT can be solved in $O^*\left(\left(2-\varepsilon_k\right)^n\right)$ time. This was first shown by Monien and Speckenmeyer in 1985 \cite{MoSp5}. Since then a long list of improvements for these $\varepsilon_k$ values have been published~\cite{Rodosek96,PPZ99,PPSZ05,schoning1999probabilistic,Hertli14a,Hertli14b,ScSt17,hansen2019faster}. Several of the most popular among the aforementioned fine-grained conjectures focus on the asymptotic behavior of these~$\varepsilon_k$ values~\cite{ImPa01,impagliazzoPaturiZane2001,vyasWilliams2021superStrong}.

The coloring problem is less well understood. The naive enumeration algorithm for $k$-coloring takes $O^*(k^n)$ time. Thus, at a first glance it is not even clear that computing the chromatic number takes ``only'' exponential time. Nonetheless, a simple dynamic-program computes the chromatic number in $O^*(3^n)$ time ~\cite{lawler1976note}. More sophisticated algorithms followed~\cite{moon1965cliques,paull1959minimizing,eppstein2001small,byskov2004enumerating}, until their culmination with a chromatic number algorithm running in $O^*(2^n)$ time by Bj{\"o}rklund, Husfeldt and Koivisto in 2006~\cite{BHK2009}.

For $k=3,4$, faster algorithms are known for the $k$-coloring problem, as well as for more general problems such as~$k$-list-coloring and even general~$(k,2)$-CSPs. Schiermeyer \cite{schiermeyer1993deciding} showed that $3$-coloring can be solved in $O^*(1.415^n)$ time. Beigel and Eppstein \cite{BeigelEppstein2005} gave algorithms solving $3$-coloring in $O^*(1.3289^n)$ time and $4$-coloring in $O^*(1.8072^n)$ time in 2005. These numbers were later improved~\cite{fomin2007improved,wu2024faster,meijer20233}. In~\cite{Zamir2021} these were extended to algorithms for~$5$-coloring and~$6$-coloring running in~$(2-\varepsilon)^n$ time, for some~$\varepsilon>0$, as well. Unlike $k$-SAT, though, for every $k>6$ the best known running time for $k$-coloring remained $O^*(2^n)$, the same as generally computing the chromatic number.

In this work, we fully resolve this long-standing gap by showing an algorithm with an improved exponent base for every fixed number of colors.
\begin{theorem*}
For every~$k\in \mathbb{N}$ there exists~$\varepsilon_k>0$ such that~$k$-coloring can be solved in~$(2-\varepsilon_k)^n$ with a randomized algorithm.
\end{theorem*}

The proof establishes the result for the more general problem of \emph{List Coloring}. In the list coloring problem, we are given a graph~$G$ and lists $L(v)$ of possible colors for each vertex $v\in V(G)$, we are asked to find a proper coloring of $G$ such that each vertex $v$ is colored by some color from $L(v)$. In the $k$-list-coloring problem each list $L(v)$ is of size at most $k$. The set~$P=\cup_v L(v)$ of all possible colors is called \emph{the palette}. Note that~$k$-coloring is a special case of~$k$-list-coloring over the palette~$[k]$ in which~$\forall v. L(v)=[k]$. The~$\Oh(2^n)$ algorithm of~\cite{BHK2009} works also for list coloring.

We extend our algorithm also to $k$-list-coloring, even when the overall color palette is of unbounded size.
\begin{theorem*}
For every~$k\in\mathbb{N}$ there exists~$\varepsilon_k>0$ such that~$k$-list-coloring over an arbitrary palette can be solved in~$\Oh((2-\varepsilon_k)^n)$ with a randomized algorithm.
\end{theorem*}

We prove the result by bringing together two complementary lines of work: the reduction from \((k+2)\)-coloring to \(k\)-list coloring developed in~\cite{Zamir2021}, and the hypergraph-container framework introduced in~\cite{Zamir2022}. The key new technical contribution we introduce is a strengthening of the former reduction to an iterable bootstrap: a sub-\(2^n\) algorithm for \(k\)-list-coloring implies one for \((k+1)\)-list-coloring, with no restriction on the palette size. Starting from polynomial-time \(2\)-list-coloring and applying the bootstrap successively gives a \(\Oh((2-\varepsilon_k)^n)\) time algorithm for every fixed maximum list size \(k\).

\paragraph{Other related works.}

Algorithms for~$k$-coloring with sub-$2^n$ running times were developed for several restricted graph families, such as bounded-degree graphs~\cite{bjorklund2010trimmed}, sparse graphs~\cite{golovnev2016families}, graphs with many low-degree vertices~\cite{Zamir2021}, and almost-regular graphs~\cite{Zamir2022}.

Recently, Bj{\"o}rklund, Curticapean, Husfeldt, Kaski, and Pratt showed that if Strassen's generalized asymptotic rank conjecture holds, then a deterministic \(\Oh(1.99982^n)\)-time algorithm for computing the chromatic number exists~\cite{BCHKP2025}. This is a very strong and non-constructive conjecture (which even in its weaker form implies that matrix multiplication admits an~$n^{2+o(1)}$ time algorithm). This conjecture is also known to be incompatible with the Set Cover conjecture~\cite{BjorklundKaski2024,Pratt2024}. Thus, the above may be interpreted as further evidence against the conjecture rather than as a conditional improved coloring algorithm, depending on the reader's beliefs.

\section{Overview}\label{sec:overview}
Our proof builds on the reduction framework introduced in~\cite{Zamir2021}. The main structural tool there, which we use as a black-box, is that for every fixed \(\alpha,\Delta>0\), List Coloring admits a sub-\(2^n\) algorithm whenever at least \(\alpha n\) vertices have degree at most \(\Delta\). Consequently, after arbitrarily choosing the constants, we may focus on graphs in which almost every vertex has large degree.

In such a graph, we may sample a small random set of vertices and color them naively (say, by enumerating over all options). That small set sees much of the graph as its neighbors, due to their high degrees. In particular, for any color that occurs many times in the neighborhood of a vertex, the random set is likely to contain a neighbor with that color, which in turn deletes this option from the vertex's list (assuming we colored the small sampled set correctly). We refer to that as the \emph{easy solver} used throughout the paper.

This reasoning gave the two reductions of~\cite{Zamir2021}: a sub-\(2^n\) algorithm for \(k\)-list-coloring implies such an algorithm for both \((k+1)\)-coloring and \((k+2)\)-coloring. The first is immediate from the above: the easy solver gets rid of at least one color for each vertex, reducing the list of color options per vertex by one. The second reduction already encounters an important obstruction. Repeatedly hitting the neighborhood of a high-degree vertex need not reveal several colors: almost all of its neighbors may receive the same color in a fixed coloring. That failure case is nevertheless useful: if a vertex has many neighbors but nearly all of them are of the same color, then we may sample a subset of its neighbors and then guess they all receive one common color and contract them. The resulting reduction, together with the previously known algorithms for $4$-list-coloring, gave the first sub-\(2^n\) algorithms for  \(5\)- and \(6\)-coloring~\cite{Zamir2021,BeigelEppstein2005}.

The first new obstacle appears when extending the argument to \((k+3)\)-coloring.  The neighborhood of a high-degree vertex may now be concentrated on two colors rather than one.  A similar reduction then does not allow us to guess that many of a vertex's neighbors are of identical color - but instead, that they must be colored with one of two specific colors. Thus, we can produce many vertices with shortened lists of size two. Treating their two possible colors by direct branching would lead to the usual~$2^n$ bound and thus consume precisely the saving we hope to obtain. On the other hand, if all vertices had lists of size two, then this would be an instance of $2$-list-coloring or $2$-SAT which is polynomial time solvable. This leads us to isolate a binary-list interpolation problem: how much faster does list-coloring become when a linear number of vertices have lists of size two?

We solve this problem by grouping vertices that share the same two-element list and separating the use of these two colors from the remaining colors.  The resulting algorithm runs in \(\Oh(2^{n-b/\binom K2})\) time when \(b\) vertices have lists of size at most two over a palette of size \(K\).  This interpolation between polynomial-time \(2\)-list-coloring and general list-coloring supplies exactly the additional saving needed for the three-step reduction.

This interpolation algorithm is particularly simple and is stated in a self-contained manner. Suppose \(b\) vertices have lists of size at most two. Ignoring technical details such as singleton lists, some pair \(Q\subseteq P\) is the list of at least \(b/\binom K2\) vertices; remove this most frequent pair class and call the remaining vertex set \(H\). Using an algorithm of~\cite{BHKK2007,BHK2009} that decides for each induced subgraph whether it is colorable in the same~$\Oh(2^n)$ time needed for solving coloring on the entire graph, we determine simultaneously which \(X\subseteq H\) are colorable using only colors outside \(Q\). For each such \(X\), the complement is restricted to the two colors of \(Q\) and is therefore solved by \(2\)-SAT in polynomial time.  The two parts of the colorings combine without any compatibility check as they are on disjoint sets of colors. The running time is
\[
   \Oh\!\left(2^{\,n-b/\binom K2}\right)
   =\Oh\!\left(2^{n-b}q_K^b\right),
   \qquad
   q_K=2^{\,1-1/\binom K2}<2.
\]
Together with the previous gap-two reduction this yields a sub-\(2^n\) algorithm for \((k+3)\)-coloring from one for \(k\)-list-coloring. In particular, the \(4\)-list-coloring algorithm of~\cite{BeigelEppstein2005} gives the first such algorithm for \(7\)-coloring.

The interpolation theorem also isolates a useful limitation of this argument. The tempting endpoint \(q=1\), under which binary-list vertices would contribute no exponential cost, is impossible assuming ETH: We construct a reduction from Traxler's bounded-frequency \((d,2)\)-CSP lower bound~\cite{Traxler2008}, by showing that such CSPs can be encoded as list-coloring instances in which most lists are of size two. On the other hand, we do not rule out a palette-independent constant \(1<q<2\). The general reduction below bypasses this interpolation question and also gives a sub-\(2^n\) algorithm for \(k\)-list-coloring over an arbitrarily large palette.

For the general reduction, the three-step warm-up appears insufficient: Binary lists have a polynomial-time endpoint, but there is no analogous reason that an instance containing many lists of size three or four should be easier than a general instance. For the full generalization then we combine the sampling reduction and the terminal algorithm instead of generalizing either one in isolation.

Consider one step of this recursive reduction: we assume a sub-\(2^n\) algorithm for \((k-1)\)-list-coloring, and aim to solve \(k\)-list-coloring, in both cases over an arbitrary palette \(P\). Unless the easy solver already applies, there are linearly many high-degree vertices whose lists have size exactly \(k\), but whose neighborhoods do not contain any sufficiently frequent color from their own lists. 
This gives us a useful primitive, similar to the one used in the reductions of~\cite{Zamir2021}: If we sample a random vertex~$v$ and a constant~$r$ neighbors of it, then with good probability all of these neighbors should be colored with colors \emph{not} appearing in~$v$'s color list~$L(v)$.
We thus sample a small linear number of vertex-disjoint stars, each consisting of a center \(v\) and a constant number \(r\) of its neighbors. With constant probability, a linear subset of the sampled stars should have all their leaves colored outside the center's color list.

We now randomly partition the palette into \(P=Q\mathbin{\dot\cup}R\), placing each color in \(R\) with probability \(1/4\). For each such star, the probability that its center's entire list lies in \(R\) and its leaves' correct colors lie in \(Q\) is at least \(4^{-k}(3/4)^r\), independently of the size of \(P\). An expectation argument gives a linear number of successful stars with constant probability, even though different stars may share colors. We cannot recognize which stars are successful, so we repeatedly guess a sufficiently small linear subcollection and restrict its leaves to \(Q\). Choosing \(r\) large enough makes the saving from these leaves outweigh the cost of guessing, while the center lists contained in \(R\) supply the other supported set. Thus, rather than merely producing many vertices with somewhat shorter lists, the reduction produces two sets with opposite palette restrictions:
\[
   \mathcal A_Q=\{v:L(v)\subseteq Q\},
   \qquad
   \mathcal B_R=\{v:L(v)\subseteq R\}.
\]
This structure replaces the size-two lists used in the warm-up.

It remains to exploit these two sets algorithmically.  Every vertex eventually colored from \(Q\) must lie outside \(\mathcal B_R\), while every vertex colored from \(R\) must lie outside \(\mathcal A_Q\). Hence we know two large, generally overlapping, supersets containing the two parts of the eventual coloring.  The difficulty is to determine how the vertices in their overlap should be divided between the two palettes, and then to combine the resulting colorings.  In the binary interpolation algorithm this combination was particularly simple, because one side was solved by \(2\)-SAT.  For larger sets \(Q\), no such argument is available.

Fortunately, the Extensions-Sum machinery developed in~\cite{Zamir2022} for the graph-container approach solves exactly this more general combination problem.  One of its main consequences is that an instance with the two supported sets above can be solved in
\[
   \Oh\!\left(
      2^{\,n-|\mathcal A_Q|}
      +2^{\,n-|\mathcal B_R|}
   \right)
\]
time.  Informally, this is comparable to solving the coloring problem separately over the possible \(Q\)-colored part and the possible \(R\)-colored part, despite the fact that their allowed vertex sets overlap. Both supported sets constructed by the reduction have linear size, so the terminal algorithm is faster than \(2^n\).

Of course, many details are swept under the rug in this high-level overview. In particular, the algorithm of course has no access to the eventual coloring and we thus cannot tell which vertices have many neighbors colored by colors from their lists and which do not. The full algorithm solves this by making a sequence of randomized guesses and setting the parameters in a way that guarantees the benefit from succeeding in such a guess outweighs the probability of failing in it.

Thus, we end up with a reduction from list-coloring to list-coloring (rather than the previous reductions from \emph{coloring} to list-coloring),
\[
   (k-1)\text{-list-coloring}
   \quad\Longrightarrow\quad
   k\text{-list-coloring}
\]
with a sub-\(2^n\) running time for every fixed \(k\geq3\), independently of the palette size. This reduction can then be bootstrapped to construct sub-$2^n$ algorithms for every~$k$.

\subsection{Organization and reading guide}
Section~\ref{sec:background} collects the black-box ingredients and the elementary reductions used later: normalization and decision-to-search, the simultaneous coloring of all induced list sub-instances of~\cite{BHKK2007}, the bounded-degree theorem of~\cite{Zamir2021}, and the two-block consequence of Extensions-Sum from~\cite{Zamir2022}. Section~\ref{sec:warmup} proves the three-step warm-up reduction, develops useful tools to be later used in the general reduction, gives the binary-list interpolation algorithm; in Appendix~\ref{sec:eth} we establish its ETH-based limitation. This section is also intended to motivate all of the objects in the general proof. Section~\ref{sec:general} proves the list-to-list bootstrap over arbitrary palettes and then iterates it to obtain the main theorem. Section~\ref{sec:discussion} is finally used for discussion and listing some remaining open problems.
In Appendix~\ref{sec:constants} we informally repeat the proof with some tedious bookkeeping to understand the asymptotic behavior of the quantitative constants our algorithm achieves; it is not necessary for proving any of our results but meant to be used as a baseline for future quantitative improvements. 

\section{Background}\label{sec:background}

The terminology used throughout the paper is standard. For a graph $G$ we denote by $V(G)$ and $E(G)$ its vertex-set and edge-set, respectively. Throughout the paper, $n$ is used to denote $|V(G)|$. For a subset $V'\subseteq V(G)$ we denote by $G[V']$ the sub-graph of $G$ induced by $V'$. For $v\in V$ we denote by $\deg(v)$ the degree of $v$ in $G$, by $N(v)$ the set of neighbors of $v$, and by $N[v]:=N(v)\cup\{v\}$. All logarithms are base two, except that $\ln$ denotes the natural logarithm.

The notation $\Oh(\cdot)$ suppresses factors polynomial in the input size. The maximum list size and all sampling parameters are fixed constants; the palette is fixed only where explicitly stated. We may discard colors that appear in no list, so a $k$-list-coloring instance has $|P|\leq kn$. We allow all algorithms to have an exponentially small error-probability. We remark that as the output of coloring algorithms is verifiable, all errors can be assumed to be one-sided.

\subsection{Coloring, List Coloring, and Constraint Satisfaction Problems}

In the \emph{$k$-coloring problem}, we are given a graph $G$ and need to decide whether there exists a $k$-coloring $c:V(G)\rightarrow [k]$ of $G$, such that for every $(u,v)\in E(G)$ we have $c(u)\neq c(v)$. If a graph has a $k$-coloring, we say that it is $k$-colorable. In the \emph{chromatic number problem}, we are given a graph $G$ and need to compute $\chi(G)$, the minimal integer $k$ for which $G$ is $k$-colorable.

For a \emph{palette} $P$, which we would usually take to be~$P=[K]$ for some integer~$K$, a list-coloring instance (over~$P$) is a graph $G$ with a nonempty list $L(v)\subseteq P$ at each vertex. A list coloring is a map $c:V(G)\to P$ satisfying $c(v)\in L(v)$ and $c(u)\neq c(v)$ on every edge~$(u,v)\in E(G)$. In the $k$-list-coloring problem (over~$P$) all lists have size at most $k$.

In a general $(a,b)$-CSP (Constraint Satisfaction Problem, see \cite{kumar1992algorithms} or \cite{schoning1999probabilistic} for a complete definition and discussions) we are given a list of \emph{constraints} on the values of subsets of size $b$ of $n$ distinct $a$-ary variables, and need to decide whether there exists an assignment of values to the variables for which all constraints are satisfied. A general constraint on a set $x_1,\ldots,x_b$ of $a$-ary variables is a subset $T$ of the $a^b$ possible assignments in $\{x_1,\ldots x_b\}\rightarrow[a]$. The constraint is satisfied by an assignment $c$, possibly on more variables, if the restriction of~$c$ to ${\{x_1,\ldots x_b\}}$ is in $T$. When~$b=2$, that is, every constraint is on a pair of variables, then without loss of generality all constraints are of the form ``$c(x_1)\neq c_1 \;\vee\;c(x_2)\neq c_2$'' for variables~$x_1,x_2$ and colors~$c_1,c_2$.

Every instance of $k$-coloring is also an instance of $k$-list-coloring (over palette $P=[k]$). Furthermore, every instance of $k$-list-coloring, irrespective of the palette size, is also a~$(k,2)$-CSP. Not to be confused with $k$-SAT, on the other hand, which is an example of a $(2,k)$-CSP.

Both $k$-coloring and $k$-list-coloring can be described as set-partition problems (using $k$ independent sets for $k$-coloring, and one allowed independent set per palette color for list-coloring); this reformulation is crucial to the~$\Oh(2^n)$ algorithm solving them both, independently of the value of~$k$~\cite{BHK2009}. The more general $(k,2)$-CSP problem though, does not have such a formulation, and indeed Traxler~\cite{Traxler2008} proved that if the Exponential Time Hypothesis (ETH) holds, then there exists a constant~$\alpha>0$ such that~$(k,2)$-CSP requires~$k^{\alpha n}$ time.

\subsection{From deciding colorability to finding a coloring}
\label{sec:decision-to-search}

The algorithms in this paper are phrased primarily as decision algorithms: they determine whether a given graph or list instance is colorable.  This does not prevent us from returning an explicit coloring.  For ordinary $k$-coloring, a standard self-reduction converts any decision algorithm into a search algorithm with only polynomial overhead; e.g., see Lemma~2.2 of~\cite{Zamir2021}.  That reduction repeatedly adds nonedges whose addition preserves $k$-colorability, until the resulting graph is edge-maximal $k$-colorable, which is simply the complement of a disjoint union of cliques, and its color classes can thus be read off directly.

In the list-coloring setting used here, there is a simpler self-reduction that does not alter the graph and only shrinks lists.  It therefore preserves the palette, the maximum list size, all degree conditions, and every other promise on the underlying graph.

\begin{lemma}[Decision-to-search for List Coloring]
\label{lem:list-decision-to-search}
Fix a palette $P$.  Suppose that list-coloring over~$P$ of an $n$-vertex instance can be decided in time $T(n)$.  Then, whenever the instance is colorable, an explicit list coloring can be found in
\[
   \Oh\bigl(T(n)\bigr)
\]
time.
\end{lemma}

\begin{proof}
Maintain a list assignment $L'$ for which the current instance is known to be colorable, initially $L'=L$.  Process the vertices one at a time. At a vertex $v$, for every $c\in L'(v)$, query the decision algorithm on the instance obtained by replacing $L'(v)$ with the singleton $\{c\}$. At least one such restriction is colorable, so an exact decision algorithm identifies a color that may be fixed at $v$.  After all vertices have been processed, every list is a singleton and these singleton colors form the required coloring. For a randomized decision algorithm, amplify each of the at most $\sum_v |L(v)|$ queries to error at most $e^{-n}/(1+\sum_v |L(v)|)$ and verify the final coloring. A union bound then gives exponentially small failure probability and polynomial overhead.
\end{proof}

Consequently, throughout the paper we may present only the decision version of each algorithm.  Whenever an explicit coloring is needed, we apply Lemma~\ref{lem:list-decision-to-search}.  Notice also that singleton restrictions never enlarge a list or expand the color palette.

\subsection{Normalization of list-coloring instances}

We use the following elementary normalization.

\begin{lemma}[Normalization]\label{lem:normalization}
Let $(G,L)$ be a list-coloring instance.
\begin{enumerate}[label=(\roman*)]
\item Deleting an edge $(u,v)$ with $L(u)\cap L(v)=\varnothing$ preserves exactly the set of list colorings.  We denote by \emph{(N)} the process of repeating this operation as long as such an edge exists.
\item Repeatedly fixing a singleton $L(v)=\{c\}$, deleting $v$, and deleting $c$ from all neighboring lists either produces an empty list, certifying infeasibility, or produces in polynomial time an equivalent residual instance with no singleton lists.  A coloring of the residual instance extends uniquely to a complete coloring of~$G$.
\end{enumerate}
Neither operation enlarges a list.
\end{lemma}

\begin{proof}
For (i), endpoints whose lists are disjoint can never be assigned the same color, so their edge imposes no constraint.  For (ii), every feasible coloring must assign $v$ the unique color $c$; properness then forbids $c$ at each neighbor.  Thus one propagation step is an equivalence, and induction proves the claim for the entire sequence.  At least one vertex is deleted at each step, so the procedure is polynomial.  Reversing these assignments reconstructs a coloring of the original instance. Because subsequent operations only shrink lists, the endpoints of every edge deleted by~(N) remain disjoint.
\end{proof}

In the warm-up we use both parts and call this \emph{full normalization}. The general proof in Section~\ref{sec:general} deliberately uses only rule~(N).

\subsection{Coloring all induced subinstances for the price of one}

The following list-coloring extension of a result by Bj\"orklund, Husfeldt, Kaski, and Koivisto~\cite{BHKK2007} will be used in the binary interpolation argument. In~\cite{BHKK2007}, the authors show that a simple modification of the $\Oh(2^n)$ time graph algorithm of~\cite{BHK2009} gives similar running time to find all $k$-colorable induced subgraphs of a given graph. While not explicitly written in their work, the same modification extends in a straightforward manner to the $\Oh(2^n)$ list-coloring algorithm from~\cite{BHK2009}. We include, briefly, this extension here for completeness.

\begin{lemma}[All induced list subinstances]\label{lem:all-subsets}
Fix a palette $P$.  Given an instance of list-coloring on a graph $G$ with lists contained in $P$, one can determine simultaneously, for every $X\subseteq V(G)$, whether $G[X]$ is list colorable, in $\Oh(2^{|V(G)|})$ time.
\end{lemma}

\begin{proof}
For each $c\in P$, define a function on subsets of $V(G)$ by
\[
 f_c(S)=1
 \quad\Longleftrightarrow\quad
 S\text{ is independent and }c\in L(v)\text{ for every }v\in S.
\]
We take $f_c(\varnothing)=1$, so a color class may be unused.  Every explicit table~$f_c(\cdot)$ is easily filled in $\Oh(2^{|V(G)|})$ time. The iterated disjoint subset convolution
\[
   F=*_{c\in P}f_c
\]
counts, for every $X$, ordered partitions of $X$ into allowed independent color classes.  Thus $F(X)>0$ exactly when $G[X]$ is list colorable. The fast subset convolution of~\cite{BHKK2007} computes all values of the convolution in $\Oh(2^{|V(G)|})$ arithmetic operations, with a factor polynomial in $|P|$ for the iterated convolutions. The counts are at most $|P|^{|V(G)|}$, so their bit lengths are polynomial in the input size. No bound on the palette size is needed.
\end{proof}

The ordinary-coloring version of Lemma~\ref{lem:all-subsets} is stated explicitly in~\cite[Section~4.3]{BHKK2007}: they observe that computing $f^{*k}(S)$ for every $S\subseteq V(G)$ finds all $k$-colorable induced subgraphs in $\Oh(2^{|V(G)|})$ total time.  The minor extension here is to use a different function $f_c$ for each color, thereby encoding the vertex lists.

\subsection{Main tool I:  Coloring with bounded-degrees}

The first of the main tools we use from previous works is a theorem from~\cite{Zamir2021}, which shows that coloring and list-coloring can be solved in sub-$2^n$ time whenever the graph has at least (any) constant fraction of vertices whose degrees are bounded by (any) constant.

\begin{definition}
For $\alpha\in[0,1]$ and $\Delta\geq0$, a graph is $(\alpha,\Delta)$-bounded if at least $\alpha |V(G)|$ of its vertices have degree at most $\Delta$.
\end{definition}

We cite the following theorem. Note that it is phrased for list-coloring and permits an arbitrary color palette.

\begin{theorem}[$(\alpha,\Delta)$-bounded List Coloring, Theorem 1.3 of~\cite{Zamir2021}]\label{thm:low-degree}
For every fixed $k,\alpha,\Delta>0$, there is a constant $C_{k,\alpha,\Delta}<2$ such that $k$-list-coloring on an $n$-vertex $(\alpha,\Delta)$-bounded graph can be solved in
\[
   \Oh\bigl(C_{k,\alpha,\Delta}^{n}\bigr)
\]
time, regardless of the size of the color palette $P$.
\end{theorem}

We use this theorem as-stated without modifications. We remark that this Theorem was derived via the introduction of a certain combinatorial subset removal lemma, which results in the advantage~$\left(2-C_{k,\alpha,\Delta}\right)$ having a rather bad dependence on the parameters $k,\alpha,\Delta$; we go back to discuss these explicit constants in Section~\ref{sec:constants}.

Theorem~\ref{thm:low-degree} was used in~\cite{Zamir2021} to construct reductions from~$(k+1)$ and~$(k+2)$ -coloring to $k$-list-coloring. In combination with the sub-$2^n$ algorithms for $3,4$-list-coloring (and in fact, even for~$(4,2)$-CSPs) of~\cite{BeigelEppstein2005} this resulted in the first sub-$2^n$ algorithms for $5,6$-coloring.

\begin{theorem}[Theorems~1.4 and~1.5 of~\cite{Zamir2021}]
\label{thm:published-gaps}
Let~$k$ be a fixed integer. For every~$\delta>0$ there is a~$\delta'>0$ such that the following hold.
\begin{enumerate}[label=(\roman*)]
\item If $(k-1)$-list-coloring is solvable in $\Oh((2-\delta)^n)$ time, then $k$-coloring is solvable in $\Oh((2-\delta')^n)$ time; this reduction is deterministic.
\item If $(k-2)$-list-coloring is solvable in $\Oh((2-\delta)^n)$ time, then $k$-coloring is solvable with exponentially small one-sided error in $\Oh((2-\delta')^n)$.
\end{enumerate}
\end{theorem}

\subsection{Main tool II: Two-block color restrictions and Extensions-Sum}

The second external tool comes from the partition-container framework of~\cite{Zamir2022}.  That work introduced algorithmic applications of the hypergraph container method, as well as a problem called Extensions-Sum in order to turn structural information supplied by graph containers into savings in inclusion-exclusion algorithms. A container may restrict an individual color class (corresponding to an independent set) to a proper subset of the vertices, but the usual inclusion-exclusion formula still contains $2^n$ terms. Extensions-Sum was introduced to exploit the fact that the contribution associated with a color depends only on its allowed vertex set.

More precisely, suppose that every color $c$ in a palette $P$ is assigned an allowed domain $V_c\subseteq V(G)$, and we ask for a proper coloring $\varphi$ satisfying
\[
   \varphi(v)=c \quad\Longrightarrow\quad v\in V_c.
\]
This is exactly a list-coloring (over palette $P$) instance under the additional restrictions
\[
   L(v)\subseteq \{c\in P:v\in V_c\}.
\]
Lemma~3.14 of~\cite{Zamir2022} expresses the relevant inclusion-exclusion computation as an Extensions-Sum instance, while Lemma~3.11 constructs its input tables.  Observation~3.23 shows that colors whose domains have small union may be merged into a single Extensions-Sum function, and Lemma~3.16 evaluates the resulting two-function instance in time equal, up to polynomial factors, to the sum of the two table sizes.  Together, these results give the following black-box statement.

\begin{theorem}[Two-block restricted coloring]\label{lem:two-block-restricted}
Let $P=Q\mathbin{\dot\cup}R$ be an arbitrary palette, and let $V_c\subseteq V(G)$ be the allowed domain of each color $c\in P$.  Define
\[
   D_Q=\bigcup_{c\in Q}V_c,
   \qquad
   D_R=\bigcup_{c\in R}V_c.
\]
Whether $G$ admits a proper coloring in which every color $c$ is used only on $V_c$ can be decided in
\[
   \Oh\bigl(2^{|D_Q|}+2^{|D_R|}\bigr)
\]
time and exponential space.
\end{theorem}

In~\cite{Zamir2022}, partition containers were constructed precisely to prove that the domains of the color classes must admit such a two-block grouping with non-trivial size bounds, leading to sub-$2^n$ coloring algorithms for regular and almost-regular graphs; see Theorems~3.24 and~3.25 therein.

For later use, we rephrase the above statement in the language of forbidden vertex sets.

\begin{corollary}[Two-block list-coloring]
\label{cor:two-block-list}
Let $P=Q\mathbin{\dot\cup}R$, and consider a list-coloring instance over a palette~$P$ on an $n$-vertex graph. Suppose that $A,B\subseteq V(G)$ satisfy
\[
   L(v)\subseteq Q \quad\text{for every }v\in A,
   \qquad
   L(v)\subseteq R \quad\text{for every }v\in B.
\]
Then list colorability can be decided in
\[
   \Oh\bigl(2^{n-|A|}+2^{n-|B|}\bigr)
\]
time and exponential space.  In particular, if $|A|,|B|\geq\eta n$, the running time is $\Oh(2^{(1-\eta)n})$.
\end{corollary}

\begin{proof}
Set $V_c=\{v:c\in L(v)\}$.  No color in $R$ is allowed on $A$, and no color in $Q$ is allowed on $B$.  Consequently,
\[
   D_R\subseteq V(G)\setminus A,
   \qquad
   D_Q\subseteq V(G)\setminus B.
\]
If a list is empty, reject immediately. Otherwise the result follows immediately from Theorem~\ref{lem:two-block-restricted}.
\end{proof}

\section{Warm-up: Extending the reduction from two to three steps}\label{sec:warmup}
As a first step, we rephrase the 2-step reductions of~\cite{Zamir2021} (from $(k+2)$ or $(k+1)$ -coloring to $k$-list-coloring) in a cleaner way that would later be useful for us in the generalization. We then extend it one step further and get a reduction from $(k+3)$-coloring to $k$-list-coloring; this already gives the first sub-$2^n$ algorithm for $7$-coloring.

Zamir's reductions follow from the~$(\alpha,\Delta)$-bounded coloring result cited as Theorem~\ref{thm:low-degree}: Given an instance of~$(k+1)$-coloring, we consider two options. If the graph is already~$(\alpha,\Delta)$-bounded for any chosen constants, then we have a sub-$2^n$ algorithm. Otherwise, nearly all vertices (a~$(1-\alpha)$ fraction, and we can take a very small constant~$\alpha>0$) have degrees that are large (larger than a constant~$\Delta$ of our choice). In that second case, we can sample a very small subset of vertices and enumerate over their correct colors. As nearly all vertices have many neighbors, this small subset is likely to hit a neighbor of the vast majority of them. In particular, most vertices lose one color option (as we already colored one of their neighbors) which puts us, modulo technical details, in a $k$-list-coloring instance.

Extending this idea to~$(k+2)$-coloring already faces an obstruction: Even if we sample a large enough subset to hit the neighborhood of each high-degree vertex more than once, it could be that in the correct coloring all of these neighbors would be assigned the same color. Thus, the resulting color lists may remain of size~$(k+1)$ and not get smaller as we sample more neighbors. When that happens for a vertex, though, it means that in the correct coloring most of its neighbors are supposed to be colored by the same color. This is useful on its own: we are able to sample many neighbors of a high-degree vertex and ``guess'' that they are all supposed to have the same color and thus can be contracted into a single vertex. As this reduces the number of vertices in the graph, a careful analysis results in the desired reduction.

Attempting to push this one step further to~$(k+3)$-coloring results in a scarier-looking obstruction, which was not yet resolved in~\cite{Zamir2021}: If high-degree vertices have their neighbors partitioned, roughly equally, between \emph{two} color classes (rather than one), then a small sampled subset of vertices will only reduce the size of lists by~$2$; at the same time, unlike the one dominant color case, knowing (or guessing) that many vertices are of one of two possible colors seems insufficient to assist a sub-$2^n$ algorithm. This is because finding the correct color among these two possible colors is still costing a factor of~$2$ for each such vertex.

Sweeping all technical details under the rug, the above sketch leaves us with a clean problem: Given a list-coloring instance in which a reasonable fraction of the vertices have only two colors in their list, can we determine colorability in sub-$2^n$ time? This sounds rather promising, as if all lists were of size two, then this would simply be an instance of $2$-SAT which can be solved in polynomial-time. It is thus reasonable to hope an interpolation between the size-two and general-size list algorithms can result in a faster algorithm in these settings.

Indeed, in Section~\ref{sec:binary} we prove that given a list-coloring instance over a palette~$P$ of fixed size~$K$, such that at least~$b$ out of the~$n$ vertices in the graph have lists of size at most two, we can determine colorability in \[ \Oh\left(2^{n-b} \cdot q^{b}\right) \] time, for $q=2^{1-1/{K \choose 2}}<2$.

At first glance, one could hope achieving the same result with~$q=1$ is possible, as that would match the natural interpolation with the polynomial-time algorithm for~$2$-SAT when~$b=n$. In Section~\ref{sec:eth} we prove that this is impossible: assuming the Exponential Time Hypothesis (ETH), any such algorithm must have~$q>1$. On the other hand, while the~$q$ we achieve depends on the palette size~$K$, we do not rule out an algorithm in which~$1<q<2$ is independent of~$K$. This gap remains open, but the general reduction in Section~\ref{sec:general} avoids it when all lists have bounded size (regardless of the palette size).

\subsection{The seed-shortening lemma}\label{sec:warmup-seed}

We begin by rehashing (and slightly generalizing) a central lemma in the reduction of~\cite{Zamir2021}, which will be useful for both the warm-up and the general algorithm.

Let $P$ be an arbitrary palette and fix $k\geq2$. Suppose $(k-1)$-list-coloring has an algorithm of base $1\leq a<2$, uniformly over arbitrary palettes, and consider a $k$-\emph{list}-coloring instance. We say a vertex is \emph{active} if its list has size exactly $k$. Relative to a fixed witness coloring $c$, call an active vertex \emph{good} if some color in its list occurs on more than $\Delta_1$ of its neighbors in the coloring~$c$, and \emph{bad} otherwise. Note that we cannot algorithmically classify vertices to good and bad as we do not know the coloring~$c$.

\begin{lemma}[Seed shortening]\label{lem:seed}
There are constants $\beta_0>0$, $\Delta_1$, and $C_S<2$ depending only on $k,a$ such that the following holds. There is an algorithm with running time $\Oh(C_S^n)$ with the following property: relative to every fixed witness coloring $c$, if at most $\beta_0n$ active vertices are bad, it returns a coloring with probability at least $1-\exp(-\Theta(n))$.
\end{lemma}

\begin{proof}
Choose every vertex independently with probability
\[
   \theta=\frac{\ln\Delta_1}{\Delta_1}
\]
to form a vertex subset~$Z\subseteq V(G)$ which we call a \emph{seed}. Abort if $|Z|>4\theta n$; otherwise enumerate all proper list-respecting colorings of $G[Z]$.  For a seed coloring $\varphi$, let $W_\varphi$ be the active vertices outside $Z$ that are not adjacent to a seed vertex colored by a color belonging to their own lists.  Abort (the enumeration on the specific~$\varphi$) if
\[
   |W_\varphi|>(\beta_0+4/\Delta_1)n.
\]
Otherwise enumerate all proper list-respecting colorings of $W_\varphi$ that are compatible with $\varphi$. Delete the explicitly colored vertices and their used colors from neighboring lists, aborting if some residual list becomes empty.  Every remaining active vertex now loses a color, so the residual instance has maximum list size $k-1$ and can be passed to the assumed algorithm for~$(k-1)$-list-coloring.

On the branch agreeing with $c$, a good active vertex is missed with probability at most
\[
   (1-\theta)^{\Delta_1}\leq \Delta_1^{-1}.
\]
The expected number of missed good active vertices is at most $n/\Delta_1$, while $\mathbb E\left[|Z|\right]=\theta n$.  Markov's inequality bounds the probability of each of the events
\[
 |Z|>4\theta n,
 \qquad
 \#\{\text{missed good active vertices}\}>4n/\Delta_1
\]
by $1/4$.  Hence, with probability at least $1/2$, neither event occurs. If there are at most $\beta_0n$ bad active vertices, then on that event the witness branch satisfies
\[
  |W_{c|_Z}|\leq(\beta_0+4/\Delta_1)n,
\]
so it is not aborted.  The enumerated coloring of $Z\cup W_{c|_Z}$ is exactly the restriction of $c$, and consequently the residual $(k-1)$-list-coloring instance is colorable.

Put
\[
   \sigma=
   4\frac{\ln\Delta_1}{\Delta_1}
   +\beta_0+\frac4{\Delta_1}.
\]
For a fixed non-aborted seed $Z$ and a seed coloring $\varphi$, the calls generated by all colorings of $W_\varphi$ cost at most
\[
  k^{|W_\varphi|}a^{n-|Z|-|W_\varphi|}.
\]
There are at most $k^{|Z|}$ seed colorings.  Since $k/a\geq1$, summing over all retained branches gives
\begin{align*}
 \sum_{\varphi}
   k^{|W_\varphi|}a^{n-|Z|-|W_\varphi|}
 &\leq k^{|Z|}a^{n-|Z|}
        (k/a)^{\max_\varphi|W_\varphi|}\\
 &\leq a^n(k/a)^{\sigma n}.
\end{align*}
Generating and checking the partial colorings is bounded by the same expression up to polynomial factors.  Thus the work over all branches is $\Oh(a^n(k/a)^{\sigma n})$. We may thus choose $\Delta_1$ large enough and then $\beta_0$ small enough so that $C_S:=a(k/a)^\sigma<2$ (this is true as we can make~$\sigma$ as small a constant as we want, and as~$a<2$).

The residual algorithm can be amplified to find a coloring with probability at least $1/2$, so one seed trial succeeds with probability at least $1/4$. Repetition amplifies this to the claimed success probability. Note that the reduction does not enlarge the palette. 
\end{proof}

Once the parameters $\Delta,\alpha$ are fixed, define the \emph{easy solver} $E$ as follows.  If at least $\alpha n$ vertices have degree at most $\Delta$, use Theorem~\ref{thm:low-degree}; otherwise use the seed solver. It has base $C<2$ and, for a suitable success probability (which can be amplified as needed by repetition), succeeds relative to a witness coloring whenever either the low-degree condition holds or at most $\beta_0n$ active vertices are bad.

\subsection{Reducing to an instance with many size-two lists}
\label{sec:pair-complement}
We now extend the reductions of~\cite{Zamir2021} by one more step, constructing a reduction from~$(k+3)$-coloring to~$k$-list-coloring, via the aforementioned size-two lists speedup which we then analyze in the next section. This section also rephrases the reductions of~\cite{Zamir2021} in a significantly more useful manner for our later generalization; instead of analyzing all three steps of the reduction simultaneously, we assume as a black-box the existing two-step reduction and analyze only the one additional step. That is, fix a palette~$P=[K]$, the previous result shows that a sub-$2^n$ algorithm for~$(K-2)$-list-coloring over~$[K]$ implies a sub-$2^n$ algorithm for~$K$-coloring (equivalently,~$K$-list-coloring over~$[K]$). Thus, it suffices to prove that a sub-$2^n$ algorithm for $(K-3)$-list-coloring over~$[K]$ implies a sub-$2^n$ algorithm for~$(K-2)$-list-coloring over~$[K]$.

Start with an instance of $(K-2)$-list-coloring over~$[K]$.  For an active vertex \(v\), that is a vertex with~$|L(v)|=(K-2)$, the complement of the list is a computable set of two colors
\[
   Q_v:=P\setminus L(v).
\]
If \(v\) is bad relative to a witness coloring, then only boundedly many neighbors of \(v\) receive colors from \(L(v)\). Consequently, when \(v\) also has sufficiently large degree, a constant sample from its neighborhood is likely to consist entirely of vertices whose witness colors lie in the pair \(Q_v\).  Restricting the sampled vertices to \(Q_v\) then preserves the witness with good probability and shrinks several lists to size at most two.

The point of taking a sample of size \(r\), rather than a single neighbor, is quantitative: if we guessed that the vertex~$v$ is bad successfully then we create \(r\) short lists. By choosing \(r\) large enough, the saving supplied by the binary-list interpolation theorem outweighs the cost of succeeding in this guess.

\begin{lemma}[Pair-complement bootstrap]\label{lem:pair-bootstrap}
Fix a palette \(P\) of size \(K\geq4\).  If \((K-3)\)-list-coloring over~$P$ has an algorithm running in \(\Oh(a^n)\) time for some \(a<2\), then \((K-2)\)-list-coloring over~$P$ has an algorithm running in \(\Oh((2-\varepsilon)^n)\) time for some \(\varepsilon>0\) that depends only on~$K,a$.
\end{lemma}

\begin{proof}
Apply Lemma~\ref{lem:seed} with \(k=K-2\), and run the full normalization (of Lemma~\ref{lem:normalization}) before the algorithm starts and after every step. Let \(\beta_0,\Delta_1,C_S\) be the constants supplied by that lemma. Set
\[
   \alpha=\frac{\beta_0}{2},
   \qquad
   M=\binom K2,
   \qquad
   p_0=\frac{\beta_0}{8}.
\]
Choose an integer \(r\geq2\) sufficiently large that
\[
   \lambda:=
   \frac{M\log_2(1/p_0)}{r}<1.
\]
Finally, put
\[
   B=(K-2)\Delta_1,
   \qquad
   \Delta=r(1+B).
\]

For a current instance on \(s\) vertices (the number of vertices may decrease from the original~$n$ due to the normalization steps), define the \emph{easy solver} \(E\) as follows.  If at least \(\alpha s\) vertices have degree at most \(\Delta\), apply Theorem~\ref{thm:low-degree}; otherwise apply the seed solver from Lemma~\ref{lem:seed}.  Both alternatives are exponential algorithms with a base smaller than \(2\).  Thus there is a constant~$C<2$ such that \(E\) runs in \(\Oh(C^s)\) time and succeeds with probability at least \(\frac12\), relative to any fixed witness, whenever either

\begin{enumerate}[label=(\roman*)]
\item at least \(\alpha s\) vertices have degree at most \(\Delta\); or
\item at most \(\beta_0s\) active vertices are bad.
\end{enumerate}

For an instance \(I\), let \(n(I)\) be its number of vertices and let \(b(I)\) be the number of vertices whose lists have size at most two.  Define
\[
   w(I):=n(I)-\frac{b(I)}M.
\]
In the following Section~\ref{sec:binary}, we show that \(I\) can be solved deterministically in \(\Oh(2^{w(I)})\) time.  We next describe a \emph{lucky move} that decreases \(w\).

Let \(I\) be a fully normalized current instance.  Choose a uniformly random vertex \(v\), and abort the move unless \(v\) is active and \(\deg(v)>\Delta\).  For such a vertex, let
\[
   Q_v=P\setminus L(v);
\]
because \(|L(v)|=K-2\), the set \(Q_v\) has size two.  Choose a uniformly random \(r\)-element subset \(T\subseteq N(v)\), and simultaneously replace
\[
   L(u)\quad\text{by}\quad L(u)\cap Q_v
   \qquad (u\in T).
\]
Reject the move if an empty list is produced; otherwise apply full normalization.  We call a move that is neither aborted nor rejected \emph{completed}.

Fix a witness coloring \(c\) of the current instance.  Suppose that neither of the two easy conditions (i) or (ii) above holds.  There are then more than \(\beta_0s\) bad active vertices, while fewer than \(\alpha s\) vertices have degree at most \(\Delta\).  Hence more than
\[
   (\beta_0-\alpha)s=\frac{\beta_0s}{2}
\]
vertices are simultaneously bad, active, and of degree greater than \(\Delta\).  The probability that the random vertex \(v\) is one of these vertices is therefore greater than \(\beta_0/2\).

\begin{figure}[t]
\centering
\begin{tikzpicture}[
    x=1cm,y=1cm,
    every node/.style={font=\small},
    center/.style={
      circle,draw=black,very thick,fill=black!8,
      minimum size=9mm,inner sep=0pt
    },
    nbr/.style={
      circle,draw=black!60,line width=.35pt,
      minimum size=3.3mm,inner sep=0pt
    },
    exceptional/.style={nbr,fill=black!15},
    qone/.style={nbr,fill=black!38},
    qtwo/.style={nbr,fill=black!75},
    edge/.style={draw=black!20,line width=.35pt},
    sample/.style={draw=black,line width=1.1pt}
]

\node[center] (v) at (0,0) {\(v\)};
\node[align=center,fill=white,inner sep=2pt] at (0,1.18) {
  \(L(v)=P\setminus Q_v\)\\[-1pt]
  \(\deg(v)>\Delta\)
};

\node[exceptional] (a11) at (-3.25, 1.05) {};
\node[exceptional] (a12) at (-2.75, 1.05) {};
\node[exceptional] (a21) at (-3.25, 0.30) {};
\node[exceptional] (a22) at (-2.75, 0.30) {};
\node at (-3.00,-0.24) {\(\vdots\)};
\node[exceptional] (ak1) at (-3.25,-0.85) {};
\node[exceptional] (ak2) at (-2.75,-0.85) {};

\node[left=4pt of a11] {\(a_1\)};
\node[left=4pt of a21] {\(a_2\)};
\node[left=4pt of ak1] {\(a_{K-2}\)};

\foreach \u in {a11,a12,a21,a22,ak1,ak2}
  \draw[edge] (v)--(\u);

\node[align=center] at (-3.00,1.72) {
  \(\leq\Delta_1\) of each\\[-1pt]
  \(a_i\in L(v)\)
};

\node[align=center] at (-3.00,-1.42) {
  at most \(B=(K-2)\Delta_1\) in total
};

\foreach \j in {1,...,5}{
  \foreach \i in {1,...,7}{
    \pgfmathsetmacro{\xx}{2.75+0.67*(\i-1)+0.13*mod(\j,2)}
    \pgfmathsetmacro{\yy}{1.40-0.70*(\j-1)}
    \pgfmathtruncatemacro{\parity}{mod(\i+\j,2)}
    \ifnum\parity=0
      \node[qone] (q-\i-\j) at (\xx,\yy) {};
    \else
      \node[qtwo] (q-\i-\j) at (\xx,\yy) {};
    \fi
    \draw[edge] (v)--(q-\i-\j);
  }
}

\foreach \i/\j in {2/1,5/2,3/3,6/4,4/5}{
  \draw[sample] (q-\i-\j) circle[radius=2.35mm];
}

\node[align=center] at (4.85,2.02) {
  Witness colors in \(Q_v=\{q_1,q_2\}\)
};

\draw[
  decorate,
  decoration={brace,mirror,amplitude=5pt}
] (2.57,-1.75)--(7.00,-1.75)
  node[midway,below=7pt] {
    \(> \Delta - B\) neighbors
  };

\end{tikzpicture}
\caption{An eligible vertex \(v\) under the witness \(c\).}
\label{fig:pair-complement-neighborhood}
\end{figure}
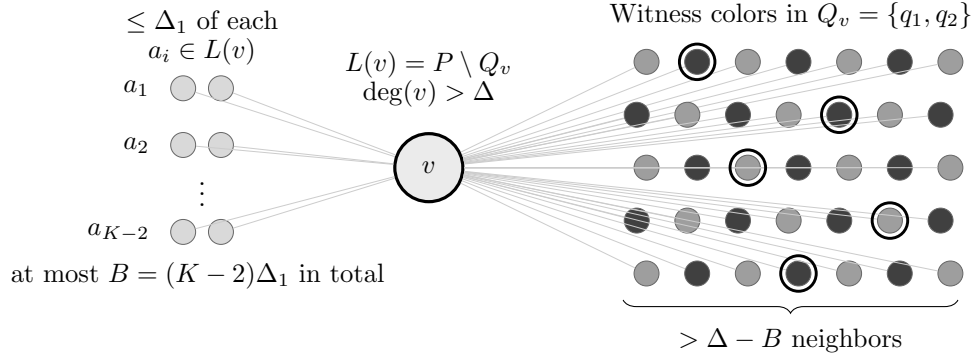

Condition on choosing such a vertex (which we call \emph{eligible}).  Since \(v\) is bad, each color in \(L(v)\) appears on at most \(\Delta_1\) of its neighbors under \(c\). Thus at most
\[
   (K-2)\Delta_1=B
\]
neighbors of \(v\) receive a witness color in \(L(v)\); every other neighbor receives a witness color in \(Q_v\).  Expose the members of the sampled set \(T\) one at a time.  At any point during the \(r\) draws, there are at least
\[
   \Delta-r\geq rB
\]
neighbors remaining.  At most \(B\) of the remaining neighbors have witness colors outside \(Q_v\). See Figure~\ref{fig:pair-complement-neighborhood} for illustration. Each draw therefore has conditional success probability at least \(1-1/r\), and
\[
 \Pr\!\left[c(T)\subseteq Q_v
       \,\middle|\, v\text{ is eligible}\right]
 \geq(1-1/r)^r\geq\frac14.
\]
Together with the probability of choosing an eligible vertex, this shows that whenever the state is not in one of the easy options (i) or (ii), a lucky move preserves the fixed witness with probability at least
\[
   \frac{\beta_0}{2}\cdot\frac14=p_0.
\]
On this event every list intersection contains the witness color, so the move is not rejected and the subsequent normalization also preserves the witness. After completing a lucky move we have
\[
   w(I')\leq w(I)-\frac rM,
\]
where \(I'\) is the normalized instance after the move.  Indeed, because rule~(N) was exhausted before the move, the edge \(uv\) implies \(L(u)\cap L(v)\neq\varnothing\) for every \(u\in T\).  Since \(Q_v\) is disjoint from \(L(v)\), intersecting \(L(u)\) with \(Q_v\) strictly shortens \(L(u)\).  On a completed move the resulting list is nonempty and has size at most two. Perform all \(r\) intersections before propagating singletons (as part of the full normalization). If \(|L(u)|\geq3\) before the move, then \(u\) becomes a short-list vertex, which decreases \(w\) by \(1/M\).  If \(|L(u)|=2\), then strict shortening makes \(u\) a singleton.  Its subsequent deletion decreases \(n(I)\) and \(b(I)\) by one, and therefore decreases \(w\) by
\[
   1-\frac1M\geq\frac1M.
\]
The vertices in \(T\) are distinct, so these \(r\) contributions add. All further singleton propagation can only decrease \(w\): deleting a short-list vertex decreases it by \(1-1/M\), while shortening a longer list to size at most two decreases it by \(1/M\).  This proves the claimed decrease in~$w$.

Let \(n\) be the number of vertices of the original input and define
\[
   d=\min\left\{\frac12,\frac{1-\log C}{2\lambda}\right\}>0,
   \qquad
   m=\left\lfloor\frac{dMn}{r}\right\rfloor.
\]
If \(m=0\), then \(n\) is bounded by a constant depending only on the fixed parameters, and the instance may be solved by exhaustive search. Assume henceforth that \(m\geq1\).

Consider the following \emph{trial} starting from the original instance and performing at most~$m$ lucky moves: At each of the \(m\) lucky move attempts, first run the easy algorithm \(E\). If \(E\) returns a coloring, return it.  Otherwise perform one lucky move; an abort or rejection ends the current trial.  After \(m\) completed lucky moves, invoke the binary interpolation (Theorem~\ref{thm:binary-interpolation}, proven in Section~\ref{sec:binary}) on the remaining instance.

To prove correctness, fix a witness coloring \(c\) of the original instance. Call a current state \emph{witness-compatible} if every remaining list contains the color assigned by the restriction of \(c\).  We claim, by (backward) induction on \(h\), that from any witness-compatible state with \(h\) lucky moves remaining, the rest of the trial succeeds with probability at least
\[
   \frac{1}{2} p_0^h.
\]
For \(h=0\), the deterministic interpolation algorithm succeeds with probability one, which is at least \(\frac{1}{2}\).  Suppose \(h\geq1\).  If the current state is easy, \(E\) succeeds with probability at least \(\frac{1}{2}\), which is at least \(\frac12 p_0^h\).  Otherwise, the lucky move preserves the witness with probability at least \(p_0\), after which the induction hypothesis applies.

In particular, one trial succeeds with probability at least \(\frac12 p_0^m\). Run
\[
   R=
   \left\lceil\frac{2n}{p_0^m}\right\rceil
\]
independent trials.  On a yes-instance the probability they all fail is
\[
   \left(1-\frac12p_0^m\right)^R
   \leq
   \exp\left(-\frac12p_0^mR\right)
   \leq e^{-n}.
\]

It remains to verify that the amplified running time is still sub-\(2^n\).  By the decrease in~$w$ at each completed move, a trial reaching its terminal call satisfies
\[
   w(I_m)
   \leq n-\frac{mr}{M}
   \leq(1-d)n+O(1).
\]
Thus the terminal call costs \(\Oh(2^{(1-d)n})\), whereas all easy calls within one trial together cost \(\Oh(C^n)=\Oh(2^{\log C \cdot n})\).  The lucky moves take only polynomial time.  Moreover,
\[
   p_0^{-m}
   \leq
   2^{(dMn/r)\log_2(1/p_0)}
   =2^{\lambda dn}.
\]
After multiplication by the number of trials, the exponential parts of the easy-call and terminal costs are therefore bounded by
\[
   2^{(\log C +\lambda d)n}
   \quad\text{and}\quad
   2^{(1-d+\lambda d)n},
\]
respectively.  By the choice of \(d\) and by \(\lambda<1\),
\[
   \log C+\lambda d
   \leq\frac{1+\log C}{2}<1,
   \qquad
   1-d+\lambda d
   =1-(1-\lambda)d<1.
\]
Hence, for
\[
   \gamma=
   \max\{\log C+\lambda d,\;1-(1-\lambda)d\}<1,
\]
the total running time is \(\Oh(2^{\gamma n})\).  Equivalently it is \(\Oh((2-\varepsilon)^n)\), where \(\varepsilon=2-2^\gamma>0\).

Every successful subroutine returns a coloring of a restriction of the original instance; forced assignments are then restored in reverse order.  Since the algorithm only shrinks lists, deletes only edges whose endpoint lists are disjoint, and verifies every returned coloring, it never accepts a no-instance.  Its only possible error is the false-negative event bounded above by~$e^{-n}$.  This proves the lemma.
\end{proof}

Thus, together with the previous reductions from~\cite{Zamir2021} we get the following.
\begin{corollary}[Three-color gap]\label{cor:gap-three}
For every fixed integer $k\geq4$, a sub-$2^n$ algorithm for $(k-3)$-list-coloring implies a sub-$2^n$ algorithm for $k$-coloring.  In particular, the known $4$-list-coloring algorithm of~\cite{BeigelEppstein2005} yields such an algorithm for $7$-coloring.
\end{corollary}

For completeness, we include a pseudo-code of the entire reduction described in Lemma~\ref{lem:pair-bootstrap} in Algorithm~\ref{alg:pair-complement}.
\vspace{11pt}

\begin{algorithm}[H]
\caption{The pair-complement bootstrap}
\label{alg:pair-complement}
\KwIn{A list-coloring instance \(I=(G,L)\) over~$[K]$ with \(|L(v)|\leq K-2\) for every \(v\)}
\KwOut{\(\mathsf{YES}\) or \(\mathsf{NO}\)}

Let \(m,R,r,\Delta\) be the constants and parameters fixed in the proof\;

\(I_0\gets\textsc{FullNormalization}(I)\)\;
\If{\(I_0\) is infeasible}{
    \Return \(\mathsf{NO}\)\;
}

\If{\(m=0\)}{
    \Return \(\textsc{BinaryInterpolation}(I_0)\)\;
}

\For{\(t\gets1\) \KwTo \(R\)}{
    \(I\gets I_0\)\;

    \For{\(j\gets1\) \KwTo \(m\)}{
        \If{\(\textsc{EasySolver-}E(I)\) finds a coloring}{
            \Return \(\mathsf{YES}\)\;
        }

        Choose \(v\) uniformly at random from \(V(I)\)\;

        \If{\(|L_I(v)|\neq K-2\) or \(\deg_I(v)\leq\Delta\)}{
            \NextTrial\;
        }

        \(Q_v\gets [K]\setminus L_I(v)\)\;
        Choose a uniformly random \(r\)-element set
        \(T\subseteq N_I(v)\)\;

        Simultaneously replace
        \(L_I(u)\) by \(L_I(u)\cap Q_v\) for every \(u\in T\)\;

        \If{one of the resulting lists is empty}{
            \NextTrial\;
        }

        \(I\gets\textsc{FullNormalization}(I)\)\;

        \If{\(I\) is infeasible}{
            \NextTrial\;
        }
    }

    \If{\(\textsc{BinaryInterpolation}(I)\) returns \(\mathsf{YES}\)}{
        \Return \(\mathsf{YES}\)\;
    }
}

\Return \(\mathsf{NO}\)\;
\end{algorithm}

\vspace{11pt}
The instruction to continue with the next trial abandons the current inner loop and returns to the outer repetition loop.

\subsection{Interpolation algorithm for instances with many size-two lists}
\label{sec:binary}

The remaining component is the black-box theorem we used in the proof of Lemma~\ref{lem:pair-bootstrap}: a faster algorithm for solving list-coloring instances over a palette of size~$K$ where many of the lists are of size (at most) two. The next theorem is stated independently of the reduction and may be useful elsewhere.

\begin{theorem}[Binary-list interpolation]\label{thm:binary-interpolation}
Fix a palette $P$ of size $K\geq2$.  If a list-coloring instance over~$P$ on $n$ vertices has $b$ vertices with lists of size at most two, it can be solved deterministically in
\[
   \Oh\!\left(2^{n-b/\binom K2}\right)
\]
time.  A coloring can be recovered within the same bound.
\end{theorem}

The algorithm we present is rooted at the following simple observation: Since there are only~$K$ colors in the palette, many of the~$b$ short-list vertices have exactly the same list of two colors. Let's call these special colors~$c_1$ and~$c_2$. If we are given a partial coloring of~$G$ in which only vertices with non-special colors in~$P\setminus\{c_1,c_2\}$ are colored -- then we can test if this partial coloring can be extended to a full coloring in polynomial time; that is because the remaining instance is simply an instance of $2$-list-coloring (or in fact, $2$-coloring as it is entirely over the palette~$\{c_1,c_2\}$). Furthermore, the test above \emph{does not depend} on the partial coloring at all - just on which vertices are colored by any color out of~$\{c_1,c_2\}$, and which are not. Hence, we can remove the vertices with the special short-list from the graph, and then use the algorithm cited in Lemma~\ref{lem:all-subsets} to find all induced subgraphs that are colorable with only colors from~$P\setminus\{c_1,c_2\}$; then attempt extending each such option to a full coloring in polynomial time.

\begin{proof}[Proof of Theorem~\ref{thm:binary-interpolation}]
First exhaust singleton propagation as part of a full normalization step. If a singleton is deleted, or if deleting its color shortens another list, the potential $n-b/\binom K2$ cannot increase, where here $b$ counts all lists of size at most two.  More explicitly, deleting a short-list vertex decreases the potential by $1-1/\binom K2$, and shortening a longer list to size at most two decreases it by $1/\binom K2$. After propagation every remaining short list has size exactly two.  It therefore suffices to prove the claim in that normalized case; relabel its parameters as $n,b$.

Let $M=\binom K2$.  Among the possible binary lists, choose a pair $Q\subseteq P$ occurring on a largest set
\[
   S=\{v:L(v)=Q\}.
\]
Then $s:=|S|\geq b/M$.  Put $H=V(G)\setminus S$.

Keep the original lists $L$ unchanged and define a separate outside-list instance on $H$ by
\[
   L_{\mathrm{out}}(v)=L(v)\setminus Q.
\]
An empty auxiliary list is allowed: it simply makes every subgraph containing that vertex infeasible in the outside-list coloring instance. By Lemma~\ref{lem:all-subsets}, in $\Oh(2^{|H|})$ total time we know, for every $X\subseteq H$, whether $G[X]$ is colorable from $L_{\mathrm{out}}$.

For every such colorable $X$, restrict each vertex of $V(G)\setminus X$ from its original list to $L(v)\cap Q$ and test the resulting instance by 2-SAT.  This is solved exactly and deterministically in polynomial time. Given a full coloring, take $X$ to be the vertices of $H$ colored outside $Q$.  Conversely, an outside-$Q$ coloring of $X$ and a $Q$-coloring of its complement combine because endpoints of an edge crossing the cut use disjoint palettes.

\begin{figure}[t]
\centering
\begin{tikzpicture}[
    x=1cm,y=1cm,
    every node/.style={font=\small},
    basevertex/.style={
      circle,draw=black!70,line width=.4pt,
      minimum size=4.1mm,inner sep=0pt
    },
    xvertex/.style={
      basevertex,fill=white,line width=.9pt
    },
    binaryvertex/.style={
      basevertex,fill=black!48
    },
    svertex/.style={
      basevertex,fill=black!48,
      double=white,double distance=.8pt,line width=.8pt
    },
    edge/.style={draw=black!25,line width=.45pt}
]

\path[fill=black!3,draw=black!60,line width=.7pt]
  (-5.90,.05)
  .. controls (-5.90,1.35) and (-5.05,1.95) .. (-3.65,1.95)
  .. controls (-2.15,1.88) and (.55,2.05) .. (1.48,1.38)
  .. controls (2.02,.92) and (2.02,-.82) .. (1.48,-1.30)
  .. controls (.55,-1.88) and (-1.80,-1.63) .. (-3.55,-1.72)
  .. controls (-5.05,-1.75) and (-5.90,-1.25) .. (-5.90,.05)
  -- cycle;

\path[fill=black!8,draw=black!60,line width=.7pt]
  (2.35,.05)
  .. controls (2.35,1.25) and (3.12,1.72) .. (4.23,1.72)
  .. controls (5.43,1.72) and (6.18,1.18) .. (6.18,.02)
  .. controls (6.18,-1.16) and (5.42,-1.67) .. (4.22,-1.67)
  .. controls (3.08,-1.67) and (2.35,-1.15) .. (2.35,.05)
  -- cycle;

\path[
  fill=white,draw=black!70,line width=.8pt,
  dash pattern=on 1.2pt off 2.1pt,line cap=round
]
  (-5.35,-.10)
  .. controls (-5.35,.48) and (-4.72,.72) .. (-3.62,.72)
  .. controls (-2.48,.72) and (-1.90,.45) .. (-1.90,-.10)
  .. controls (-1.90,-.79) and (-2.45,-1.08) .. (-3.55,-1.08)
  .. controls (-4.72,-1.08) and (-5.35,-.79) .. (-5.35,-.10)
  -- cycle;

\coordinate (x1) at (-4.82,-.05);
\coordinate (x2) at (-4.02,.23);
\coordinate (x3) at (-3.05,.02);
\coordinate (x4) at (-4.45,-.63);
\coordinate (x5) at (-3.46,-.54);
\coordinate (x6) at (-2.40,-.48);

\coordinate (h1) at (-1.08,.18);
\coordinate (h2) at (-.10,.40);
\coordinate (h3) at (.92,.18);
\coordinate (h4) at (-1.05,-.68);
\coordinate (h5) at (.02,-.86);
\coordinate (h6) at (1.10,-.55);

\coordinate (s1) at (3.04,.05);
\coordinate (s2) at (4.10,.34);
\coordinate (s3) at (5.22,.04);
\coordinate (s4) at (3.38,-.70);
\coordinate (s5) at (4.48,-.55);
\coordinate (s6) at (5.48,-.76);

\foreach \u/\w in {
  x1/x2,x1/x4,x2/x3,x2/x5,x3/x5,x3/x6,x4/x5,x5/x6,
  x2/h1,x3/h1,x3/h2,x5/h4,x6/h1,x6/h4,
  h1/h2,h1/h4,h2/h3,h2/h5,h3/h6,h4/h5,h5/h6,
  h2/s1,h3/s2,h3/s4,h5/s1,h6/s4,h6/s5,
  s1/s2,s1/s4,s2/s3,s2/s5,s3/s5,s3/s6,s4/s5,s5/s6}{
    \draw[edge] (\u)--(\w);
}

\foreach \u in {x1,x2,x3,x4,x5,x6}
  \node[xvertex] at (\u) {};

\foreach \u in {h1,h2,h3,h4,h5,h6}
  \node[binaryvertex] at (\u) {};
\foreach \u in {s1,s2,s3,s4,s5,s6}
  \node[svertex] at (\u) {};

\node[font=\small\bfseries] at (-2.05,2.30) {
  \(H=V(G)\setminus S\)
};
\node[font=\small\bfseries] at (4.25,2.12) {
  \(S=\{v:L(v)=Q\}\)
};

\node[align=center] at (-3.62,1.25) {
  \(X\subseteq H\)\\[-1pt]
  \(L_{\mathrm{out}}(v)=L(v)\setminus Q\)
};
\node[align=center] at (0,1.25) {
  \(H\setminus X\)\\[-1pt]
  lists \(L(v)\cap Q\)
};
\node[align=center] at (4.25,1.22) {
  Every list is \(Q\)
};

\end{tikzpicture}
\caption{The decomposition used by the binary-list interpolation algorithm.}
\label{fig:binary-interpolation-nested}
\end{figure}
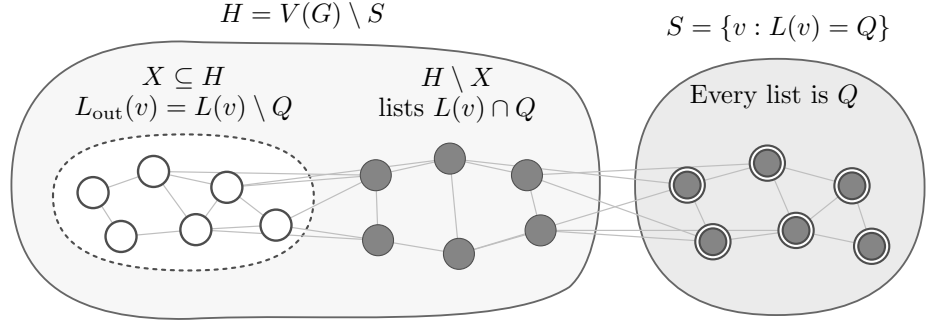

There are $2^{|H|}=2^{n-s}$ subsets and polynomial work per subset.  Thus the running time is at most
\[
   \Oh(2^{n-s})
   \leq \Oh\!\left(2^{n-b/M}\right).
\]
After finding a successful $X$, we can also find an explicit coloring of it in~$\Oh(2^{|X|})$ time.  This reconstructs the outside-$Q$ coloring in at most $\Oh(2^{|H|})$ additional time; 2-SAT returns the complementary coloring. Thus a full coloring is recovered within the same bound.
\end{proof}

For completeness, we include the pseudo-code in Algorithm~\ref{alg:binary-interpolation} and an illustration of the decomposition in Figure~\ref{fig:binary-interpolation-nested}.
\vspace{11pt}

\begin{algorithm}[H]
\caption{Binary-list interpolation}
\label{alg:binary-interpolation}
\KwIn{A list-coloring instance \(I=(G,L)\) over a fixed palette \(P\)}
\KwOut{\(\mathsf{YES}\) if \(I\) is list colorable, and
       \(\mathsf{NO}\) otherwise}

\(I\gets\textsc{FullNormalization}(I)\)\;

\If{\(I\) is infeasible}{
    \textbf{return} \(\mathsf{NO}\)\;
}

Choose a pair \(Q\in\binom{P}{2}\) maximizing
\(
   \bigl|\{v\in V(G):L(v)=Q\}\bigr|
\).

Set
\(
   S\gets\{v\in V(G):L(v)=Q\},
   \;
   H\gets V(G)\setminus S.
\)

For every \(v\in H\), set
\(
   L_{\mathrm{out}}(v)\gets L(v)\setminus Q
\).

\(\mathcal A\gets
  \textsc{AllInducedListSubinstances}
  (G[H],L_{\mathrm{out}})\)\;

\For{each \(X\subseteq H\) with
     \(\mathcal A[X]=\mathsf{YES}\)}{
    For every \(v\in V(G)\setminus X\), set
    \(
       L_Q(v)\gets L(v)\cap Q.
    \)

    \If{\(\textsc{2SAT}
          (G[V(G)\setminus X],L_Q)=\mathsf{YES}\)}{
        \textbf{return} \(\mathsf{YES}\)\;
    }
}

\textbf{return} \(\mathsf{NO}\)\;
\end{algorithm}

\vspace{11pt}
The running time in Theorem~\ref{thm:binary-interpolation} can equivalently be written as
\[
   \Oh\!\left(2^{n-b}q_K^b\right),
   \qquad
   q_K=2^{\,1-1/\binom K2}<2.
\]
Thus a vertex with a general list contributes a factor of~$2$, whereas a vertex with a list of size at most two contributes the smaller factor $q_K$.  Since an instance in which every list has size at most two is solvable in polynomial time, it is natural to ask whether the endpoint $q=1$ can be attained.  This would correspond to both desired endpoints of the interpolation: at~$b=0$, we simply solve a standard list-coloring instance in~$\Oh(2^n)$ time, and at~$b=n$ we solve a $2$-list-coloring instance in polynomial time.

In Appendix~\ref{sec:eth} we show that this optimistic endpoint $q=1$ is impossible under the Exponential Time Hypothesis. In fact, our reduction rules out a whole interval of constants larger than one.  The palette used by the reduction depends on one fixed CSP domain size, but not on the number of variables.  Consequently, the lower bound applies even if the hidden constants in the list-coloring algorithm are allowed to depend arbitrarily on the fixed palette.

\subsection{Reflection time}

This section provides a good understanding of both the technical basis for the general algorithm as well as the limitation of the binary-list interpolation approach.

In terms of generalization, at first glance we may be discouraged from extending the present reduction further: While a list-coloring instance with many lists of size two seems like a naturally easier problem, due to the polynomial-time algorithm for the extreme case of only such size two lists, the same is no longer true for larger lists. There is no apriori reason to believe that an instance with many lists of size three or four is easier than a general instance. A glimpse of hope though emerges from the recursive structure of our constructions. If we already have a sub-$2^n$ algorithm for $k$-list-coloring over~$[K]$, we may hope that an instance with larger lists in which many lists are guaranteed to be of sizes at most~$k$ might still have a faster solution. Materializing this hope is not straightforward though, and a combination of tools appearing in the size-two list interpolation algorithm with tools from the hypergraph container approach in~\cite{Zamir2022} is needed to do so. We do that in Section~\ref{sec:general}.

As for the obstruction, we make the following observation. The interpolation constant $q_K=2^{\,1-1/\binom K2}<2$ we obtain in Section~\ref{sec:binary} approaches two as $K\rightarrow\infty$. It remains open whether some palette-independent $1<q<2$ is possible. The general reduction below bypasses this dependence: when \emph{every} list has bounded size (and not just the short ones), we present a different savings mechanism. Thus, the palette-independent interpolation question remains open, but is not needed for $k$-list-coloring over arbitrary palettes. In that interpolation question the other lists may have arbitrarily large sizes, whereas our algorithm requires a fixed bound on every list.

\section{The general algorithm}
\label{sec:general}

We aim to generalize the reduction of Section~\ref{sec:warmup} in a natural way: we construct a reduction from a sub-$2^n$ algorithm for~$(k-1)$-list-coloring to a sub-$2^n$ algorithm for~$k$-list-coloring, for every fixed~$k\geq3$, even over an arbitrary palette~$P$. In the warm-up, we constructed such reductions only for~$K-2\leq k\leq K$ over a fixed palette~$[K]$.

Cumbersome (and many) technical details aside, the warm-up reduction can be described as follows: Unless the instance has an `easy' structure we already have a fast solution for, we find in it many vertices~$v$ that have lists~$L(v)$ of the maximum possible size~$k=(K-2)$ and that also have many distinct neighbors whose correct color must be in the smaller complement lists~$Q_v=[K]\setminus L(v)$. There, the size of these complements is~$|Q_v|=2$.
Then, we presented an algorithm to solve such instances with many size-two lists.

While tempting to generalize just the first part and again encapsulate the second part into some natural black-box statement, we are unable to do so. Intuitively, it truly is not clear if list-coloring instances with many lists of size three or four are easier to solve than general ones. Instead, we take inspiration, or parts, from \emph{both} components at the same time to construct the general reduction.

As before, the only case not solved by the `easy' routines is when we have many vertices~$v$ for which nearly all of their neighbors should be colored with a color \emph{outside} of the list~$L(v)$.
Using random guessing, we construct a large set of vertices with maximum-size lists~$L(v)$ and sample a large yet constant subset of their neighbors, we keep these sampled stars vertex-disjoint. 
With constant probability, linearly many stars from the aforementioned \emph{constellation} have their leaves colored outside their center's list in a fixed witness coloring. 

We then partition the palette randomly into two sets~$Q,R$. As every center list has bounded size, we expect a constant fraction of the constellation to have their center lists contained in~$R$ and their leaves' witness colors in~$Q$. 
Guessing a small enough subcollection of the constellation and restricting their leaves to~$Q$ gives us two disjoint sets of colors such that each is avoided by a known set of linearly many vertices, which is the desired property to use the algorithm of Theorem~\ref{lem:two-block-restricted}. 

The cost of guessing is paid for by taking sufficiently many leaves per star. Now, we have a graph with many vertices that are forbidden from using colors in~$\emptyset \neq Q\subsetneq P$, and also many vertices that are forced to use only colors from~$Q$. This now sounds reminiscent of the situation we tackle in the interpolation algorithm of Section~\ref{sec:binary}: There are~$(1-\varepsilon)n$ known vertices (for some~$\varepsilon>0$) which are a superset of the part of the graph that would eventually be colored by colors of~$P\setminus Q$, and similarly, there are~$(1-\varepsilon)n$ vertices which are a superset of the part that would be colored by colors of~$Q$. Unlike Section~\ref{sec:binary} though, ``stitching'' the two parts together -- or finding which colorable subgraphs of each are compatible with each other -- is more difficult. Fortunately, the machinery in~\cite{Zamir2022} solves \emph{exactly} this problem; as discussed and cited in Theorem~\ref{lem:two-block-restricted}, we \emph{are} able to solve that ``stitching'' problem in time comparable to enumerating over subsets of the two super-sets separately.

\begin{figure}[hb]
\centering
\begin{tikzpicture}[
  x=1cm,
  y=1cm,
  every node/.style={text=black},
  outer boundary/.style={
    draw=black!72,
    fill=white,
    line width=0.85pt
  },
  set boundary/.style={
    draw=black!68,
    line width=0.70pt
  },
  vertex/.style={fill=black!60}
]

\def\outerpath{%
  (-3.08,-0.28)
  .. controls (-3.03,-1.13) and (-2.26,-1.53) .. (-1.33,-1.52)
  .. controls (-0.49,-1.66) and (0.37,-1.47) .. (1.18,-1.55)
  .. controls (2.28,-1.56) and (3.01,-1.03) .. (3.08,-0.18)
  .. controls (3.21,0.76) and (2.43,1.43) .. (1.43,1.49)
  .. controls (0.53,1.62) and (-0.31,1.44) .. (-1.17,1.53)
  .. controls (-2.23,1.55) and (-3.03,0.72) .. cycle}

\newcommand{\vertexset}{%
  \foreach \p in {
    (-2.66,-0.18),(-2.45,0.48),(-2.36,-0.76),(-2.15,0.88),
    (-1.42,0.22),(-1.12,-0.75),(-0.72,0.72),(-0.36,-0.08),
    (-0.12,0.95),(0.22,-0.85),(0.62,0.34),(1.05,-0.32),
    (2.15,-0.88),(2.36,0.76),(2.45,-0.48),(2.66,0.18)}
    \fill[vertex] \p circle (1.45pt);
}

\begin{scope}[shift={(-4.18,0)}]
  \path[outer boundary] \outerpath;

  \begin{scope}
    \clip \outerpath;

    \path[fill=black!8,even odd rule]
      \outerpath
      (1.05,-0.01) ellipse [x radius=3.18,y radius=2.03];
    \path[fill=black!27,even odd rule]
      \outerpath
      (-1.05,-0.01) ellipse [x radius=3.18,y radius=2.03];

    \draw[draw=black!72,line width=0.75pt,densely dotted]
      (1.05,-0.01) ellipse [x radius=3.18,y radius=2.03];
    \draw[draw=black!68,line width=0.75pt,dashed]
      (-1.05,-0.01) ellipse [x radius=3.18,y radius=2.03];
  \end{scope}

  \vertexset
  \path[draw=black!72,line width=0.85pt] \outerpath;

  \node[font=\small\bfseries] at (0,2.02)
    {(a) Supported sets inside \(V(G)\)};

  \path[set boundary,fill=black!8]
    (-2.72,-2.07) rectangle ++(0.25,0.25);
  \node[anchor=west,font=\small] at (-2.35,-1.945)
    {\(\mathcal A_Q:\ L(v)\subseteq Q\)};

  \path[set boundary,fill=black!27]
    (0.40,-2.07) rectangle ++(0.25,0.25);
  \node[anchor=west,font=\small] at (0.77,-1.945)
    {\(\mathcal B_R:\ L(v)\subseteq R\)};
\end{scope}

\begin{scope}[shift={(4.18,0)}]
  \path[outer boundary] \outerpath;

  \begin{scope}
    \clip \outerpath;

    \path[fill=black!8]
      (-1.05,-0.01) ellipse [x radius=3.18,y radius=2.03];

    \path[fill=black!31,fill opacity=0.50]
      (1.05,-0.01) ellipse [x radius=3.18,y radius=2.03];

    \draw[draw=black!68,line width=0.75pt,dashed]
      (-1.05,-0.01) ellipse [x radius=3.18,y radius=2.03];
    \draw[draw=black!72,line width=0.75pt,densely dotted]
      (1.05,-0.01) ellipse [x radius=3.18,y radius=2.03];
  \end{scope}

  \vertexset
  \path[draw=black!72,line width=0.85pt] \outerpath;

  \node[font=\small\bfseries] at (0,2.02)
    {(b) Overlapping complements inside \(V(G)\)};

  \draw[draw=black!68,line width=0.75pt,dashed]
    (-2.72,-1.95)--(-2.36,-1.95);
  \node[anchor=west,font=\small] at (-2.25,-1.95)
    {\(V(G)\setminus\mathcal B_R\)};
  \node[font=\scriptsize,text=black!65] at (-1.48,-2.28)
    {\(Q\)-side};

  \draw[draw=black!72,line width=0.75pt,densely dotted]
    (0.30,-1.95)--(0.66,-1.95);
  \node[anchor=west,font=\small] at (0.77,-1.95)
    {\(V(G)\setminus\mathcal A_Q\)};
  \node[font=\scriptsize,text=black!65] at (1.56,-2.28)
    {\(R\)-side};
\end{scope}

\draw[->,line width=0.75pt]
  (-0.50,0)--(0.50,0);
\node[font=\scriptsize] at (0,0.32) {complements};

\node[font=\small] at (0,2.70)
  {\(R=P\setminus Q\)};

\end{tikzpicture}
\caption{The same two boundaries in two views.  The supported sets
\(\mathcal A_Q\) and \(\mathcal B_R\) are exactly the two differences
of the overlapping complementary domains shown on the right.}
\label{fig:general-two-block-overview}
\end{figure}
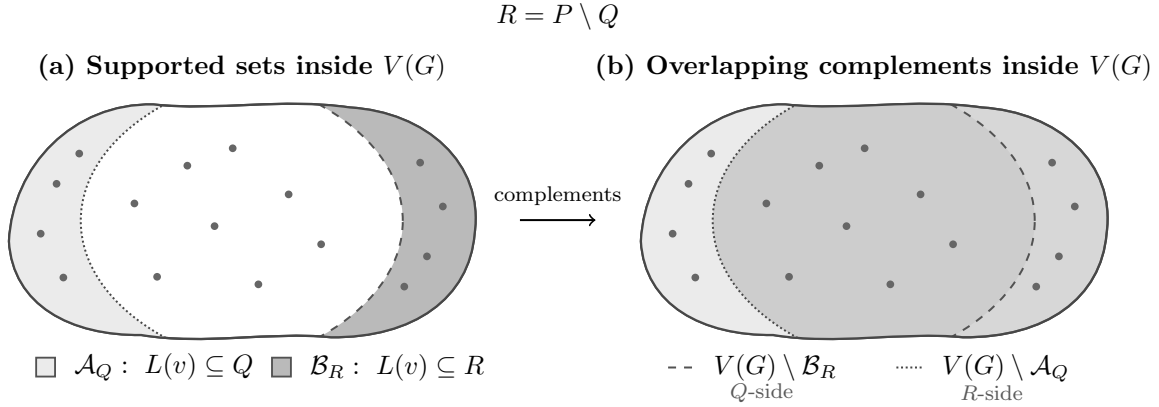

\subsection{The list-to-list bootstrap}

We now make the preceding outline precise. As in the warm-up, a vertex is \emph{active} when its list has the maximum currently allowed size, and goodness and badness are always defined relative to a fixed witness coloring. 

\begin{theorem}[General bootstrap]\label{thm:general-bootstrap}
Fix an integer \(k\geq3\). If \((k-1)\)-list-coloring over arbitrary palettes has a randomized algorithm running in \(\Oh(a^n)\) time for some \(a<2\), then \(k\)-list-coloring over arbitrary palettes has a randomized algorithm running in
\[
   \Oh\bigl((2-\varepsilon)^n\bigr)
\]
time for some \(\varepsilon>0\) depending only on \(k,a\). 
\end{theorem}

\begin{proof}
Apply Lemma~\ref{lem:seed} with the present values of \(k,a\), and let \(\beta_0,\Delta_1,C_S\) be the constants it supplies, taking \(\beta_0\leq1\). Put
\[
   \alpha:=\frac{\beta_0}{4},
   \qquad
   \rho:=\frac{\beta_0}{16}.
\]
Choose an integer \(r\geq2\) sufficiently large that
\[
   r\log_2(3/2)>2k+\log_2(4/\rho),
\]
and set
\[
   B:=k\Delta_1,
   \qquad
   \Delta:=r(1+B).
\]
Let \(1\leq C_{\mathrm{LD}}<2\) be the base supplied by Theorem~\ref{thm:low-degree} for these \(k,\alpha,\Delta\). Choose \(\delta>0\) sufficiently small that
\[
   \delta\leq\frac{\beta_0}{4},
   \qquad
   k^\delta C_{\mathrm{LD}}^{1-\delta}<2.
\]
This order is important: the sampling size \(r\) is fixed before the degree threshold \(\Delta\), and the fraction \(\delta\) is chosen after the bounded-degree algorithm's base is known.

Finally, define
\[
   \gamma:=\frac{\rho}{2}\,4^{-k}(3/4)^r,
   \qquad
   \kappa:=r-\log_2(2/\gamma)
          =r\log_2(3/2)-2k-\log_2(4/\rho)>0,
\]
and, for an input with \(n\) vertices, put
\[
   m:=\left\lfloor\frac{\delta n}{r+1}\right\rfloor,
   \qquad
   \ell:=\left\lfloor\frac{\gamma m}{4r}\right\rfloor.
\]
If \(\ell=0\), then \(n\) is bounded by a constant depending only on the fixed parameters, and the instance may be solved by exhaustive search. We henceforth assume \(\ell\geq1\). Reject an empty list immediately and exhaust rule~(N) from Lemma~\ref{lem:normalization}. Singleton vertices are retained. Every trial described below starts from a fresh copy of this instance, and every subroutine call and sample uses fresh independent randomness.

First run the seed solver on a copy of the input. If it finds a coloring, return it. Otherwise construct a constellation of \(m\) vertex-disjoint stars as follows. Maintain a set \(W\) of vertices used in previous stars, initially empty, and let
\[
   H:=G[V(G)\setminus W].
\]
Before each step, if at least \(\alpha n\) vertices have degree at most \(\Delta\) in \(H\), enumerate all proper list-respecting colorings of \(G[W]\). For each such coloring, on a fresh copy of the lists on \(H\), delete its colors from the neighboring lists and invoke the bounded-degree solver, skipping branches with empty lists. If a coloring is found, extend it by the coloring of \(W\) and return it; otherwise end the trial. We call this the \emph{low-degree exit}.

The low-degree exit is valid because the \(\alpha n\) low-degree vertices form at least an \(\alpha\)-fraction of \(V(H)\), and fixing the colors of \(W\) does not increase their degrees or list sizes. Moreover, \(|W|\leq\delta n\), so its cost is at most
\[
   \Oh\!\left(k^{|W|}C_{\mathrm{LD}}^{n-|W|}\right)
   \leq
   \Oh\!\left((k^\delta C_{\mathrm{LD}}^{1-\delta})^n\right).
\]
Thus it has base smaller than two. On a colorable instance the enumerated branch agreeing with a witness is retained, and the solver succeeds with probability at least \(1/2\), after constant amplification if needed.

If the low-degree exit does not apply, choose uniformly an active vertex \(v\in V(H)\) with \(\deg_H(v)>\Delta\), and choose a uniformly random \(r\)-element set \(T\subseteq N_H(v)\). If there is no such vertex, end the trial. Otherwise we add \((v,T)\) to our constellation and put
\[
   W\leftarrow W\cup\{v\}\cup T.
\]
The stars are vertex-disjoint by construction, although edges among their leaves or between different stars are allowed.

We next analyze this constellation relative to a fixed witness coloring \(c\). If at most \(\beta_0n\) active vertices are bad, the first call to the seed solver already succeeds with constant probability. Suppose therefore that more than \(\beta_0n\) active vertices are bad. Call a sampled star \((v,T)\) \emph{clean} if
\[
   c(T)\cap L(v)=\varnothing.
\]
As long as the low-degree exit has not occurred, more than
\[
   \beta_0n-|W|-\alpha n\geq\frac{\beta_0n}{2}
\]
of the original bad active vertices remain unused and have degree greater than \(\Delta\) in \(H\). Their lists have not changed, and deleting vertices cannot increase the number of neighbors with witness color in any one color. Hence an eligible center exists, and the probability that the chosen center is bad is at least \(\beta_0/2\).

Condition on choosing such a center. Since it is bad, at most \(k\Delta_1=B\) of its remaining neighbors have witness color in \(L(v)\). Expose the \(r\) sampled neighbors one at a time. Before each draw at least \(\Delta-r=rB\) neighbors remain available, so each draw avoids these exceptional neighbors with probability at least \(1-1/r\). Thus
\[
   \Pr\!\left[c(T)\cap L(v)=\varnothing
       \,\middle|\,v\text{ is bad}\right]
   \geq(1-1/r)^r\geq\frac14.
\]
Consequently, conditional on every preceding constellation, the next star is clean with probability at least \(\beta_0/8=2\rho\), unless the low-degree exit occurs.

To account for that possible exit, give each completed step score one if its star is clean and zero otherwise. If the low-degree exit occurs, give that step and all remaining steps score one. The total score \(Y\) satisfies
\[
   \mathbb E[Y]\geq2\rho m,
   \qquad 0\leq Y\leq m,
   \qquad \Pr[Y\geq\rho m]\geq\rho.
\]
The last inequality follows from \(\mathbb E[Y]\leq\rho m+m\Pr[Y\geq\rho m]\). Therefore, with probability at least \(\rho\), either the low-degree exit occurs or the final constellation has at least \(\rho m\) clean stars. 

Suppose the constellation is completed, and write its stars as \((v_i,T_i)\), \(1\leq i\leq m\). Independently put each color of \(P\) in \(R\) with probability \(1/4\), and in \(Q\) with probability \(3/4\). Call a clean star \emph{successful} if
\[
   L(v_i)\subseteq R,
   \qquad c(T_i)\subseteq Q.
\]
For a clean star the two sets of colors in these requirements are disjoint, so
\[
   \Pr[\text{star }i\text{ is successful}]
   =4^{-|L(v_i)|}(3/4)^{|c(T_i)|}
   \geq4^{-k}(3/4)^r.
\]
If at least \(\rho m\) stars are clean, the number \(X\) of successful stars therefore satisfies \(\mathbb E[X]\geq2\gamma m\) and \(0\leq X\leq m\). The same expectation argument gives
\[
   \Pr[X\geq\gamma m]\geq\gamma.
\]
Note that we didn't assume or use independence between stars, as their lists and witness colors may overlap arbitrarily.

For the sampled partition, compute the set~$\mathcal B_R$ of all centers~$v$ of stars in the constellation for which $L(v)\subseteq R$.
End the trial if \(|\mathcal B_R|<r\ell\). Otherwise choose uniformly an \(\ell\)-element subset \(J\subseteq[m]\), let
\[
   A:=\bigcup_{i\in J}T_i,
\]
and, on a fresh copy of the input, restrict \(L(u)\) to \(L(u)\cap Q\) for every \(u\in A\). Skip this selection if an empty list is produced; otherwise invoke Corollary~\ref{cor:two-block-list}.

We first verify the deterministic saving for every non-aborting selection. The stars are disjoint, so \(|A|=r\ell\). We have \(A\cap\mathcal B_R=\varnothing\), and all the lists on \(\mathcal B_R\) remain unchanged. The two-block terminal consequently runs in
\[
   \Oh\bigl(2^{n-|A|}+2^{n-|\mathcal B_R|}\bigr)
   =\Oh(2^{n-r\ell}).
\]
This is a worst-case bound and does not rely on the selection having preserved \(c\).

On the event \(X\geq\gamma m\), the distinct centers of successful stars show that \(|\mathcal B_R|\geq\gamma m\geq r\ell\), so the cardinality check succeeds. Since \(\ell\leq\gamma m/(4r)\), a uniform selection of \(\ell\) stars consists entirely of successful stars with probability at least
\[
   \prod_{j=0}^{\ell-1}\frac{X-j}{m-j}
   \geq(\gamma/2)^\ell.
\]
Such a selection preserves the witness. Make
\[
   N_{\mathrm{sel}}:=\left\lceil3(2/\gamma)^\ell\right\rceil
\]
independent selections for this partition, returning a coloring if any succeeds. Conditional on \(X\geq\gamma m\), the probability that all fail is at most \(e^{-3}\). Their total cost is
\[
   \Oh\!\left((2/\gamma)^\ell 2^{n-r\ell}\right)
   =\Oh(2^{n-\kappa\ell}).
\]
This is why we chose a biased palette partition: the cost \(4^k\) is paid once per selected star, whereas its \(r\) leaves each contribute a saving, and the factor \((4/3)^r\) from guessing their colors' side is smaller than \(2^r\).

Combining the constellation and palette probabilities, one entire trial on a hard colorable instance succeeds with probability at least \(\rho\gamma/2\). Indeed, a low-degree exit succeeds with probability at least \(1/2\); a completed packing with \(\rho m\) clean stars succeeds with probability at least \(\gamma(1-e^{-3})\geq\gamma/2\). If the seed condition holds instead, the first call succeeds with probability at least \(1/2\), which is also sufficient. Repeat the entire trial independently
\[
   N_{\mathrm{tr}}:=\left\lceil\frac{2n}{\rho\gamma}\right\rceil
\]
times. The probability of missing a coloring is at most \(e^{-n}\).

It remains to combine the running times. The seed solver and low-degree exit are invoked at most once per trial. Only the two-block calls are repeated \(N_{\mathrm{sel}}\) times. Since
\[
   \ell=\frac{\gamma\delta}{4r(r+1)}n+O(1),
\]
the total cost is bounded by \(\Oh((2-\varepsilon)^n)\), where
\[
   2-\varepsilon
   :=\max\left\{
       C_S,\ k^\delta C_{\mathrm{LD}}^{1-\delta},\quad
       2^{\,1-\kappa\gamma\delta/(4r(r+1))}
     \right\}<2.
\]
All other work, including the \(N_{\mathrm{tr}}\) outer trials, contributes only polynomial factors for fixed \(k,a\). Every acceptance is sound: colorings returned by subroutines are verified, fixing colors in \(W\) is checked against the original lists and edges, and the terminal instances only shrink lists. This proves the theorem.
\end{proof}

For completeness, Algorithm~\ref{alg:general-bootstrap} gives the entire reduction. Calls to the seed solver, the bounded-degree solver, and the two-block algorithm are left as the named black boxes already analyzed above; the sampling and list modifications are shown explicitly.

\vspace{11pt}
\begin{algorithm}[H]
\caption{The list-to-list bootstrap}
\label{alg:general-bootstrap}
\KwIn{A list-coloring instance \(I=(G,L)\) over an arbitrary palette \(P\), with \(|L(v)|\leq k\) for every \(v\)}
\KwOut{\(\mathsf{YES}\) or \(\mathsf{NO}\)}
Let the constants and repetition counts be those in Theorem~\ref{thm:general-bootstrap}\;
\lIf{a list is empty}{\Return \(\mathsf{NO}\)}
\(I_0\gets\textsc{Normalize-(N)}(I)\); \((G,L)\gets I_0\)\;
\lIf{\(\ell=0\)}{\Return \(\textsc{ExhaustiveListColoring}(I_0)\)}
\For{\(j\gets1\) \KwTo \(N_{\mathrm{tr}}\)}{
    \lIf{\(\textsc{SeedSolver}(I_0)\) finds a coloring}{\Return \(\mathsf{YES}\)}
    \(W\gets\varnothing\)\;
    \For{\(i\gets1\) \KwTo \(m\)}{
        \(H\gets G[V(G)\setminus W]\)\;
        \If{at least \(\alpha n\) vertices of \(H\) have degree at most \(\Delta\)}{
            Enumerate proper list-colorings of \(G[W]\), propagate them, and run the bounded-degree solver on each feasible residual instance on \(H\)\;
            \lIf{an extension is found}{\Return \(\mathsf{YES}\)}
            \NextTrial\;
        }
        \lIf{no active vertex of \(H\) has degree greater than \(\Delta\)}{\NextTrial}
        Choose such a vertex \(v_i\) uniformly, and a uniform \(r\)-subset \(T_i\subseteq N_H(v_i)\)\;
        \(W\gets W\cup\{v_i\}\cup T_i\)\;
    }
    Put each color in \(R\) independently with probability \(1/4\); set \(Q\gets P\setminus R\)\;
    \(\mathcal B_R\gets\{v:L(v)\subseteq R\}\)\;
    \lIf{\(|\mathcal B_R|<r\ell\)}{\NextTrial}
    \For{\(h\gets1\) \KwTo \(N_{\mathrm{sel}}\)}{
        Choose a uniform \(\ell\)-subset \(J\subseteq[m]\); set \(A\gets\bigcup_{i\in J}T_i\)\;
        \(I'\gets I_0\); replace \(L_{I'}(u)\) by \(L(u)\cap Q\) for all \(u\in A\)\;
        \If{no list is empty}{
            \lIf{\(\textsc{TwoBlockListColoring}(I';Q,R)\) returns \(\mathsf{YES}\)}{\Return \(\mathsf{YES}\)}
        }
    }
}
\Return \(\mathsf{NO}\)\;
\end{algorithm}
\vspace{11pt}

We finally deduce as a corollary our main theorem.

\begin{theorem}[$k$-Coloring and List Coloring]\label{thm:main}
For every fixed integer \(k\geq1\), there is an \(\varepsilon_k>0\) such that \(k\)-coloring can be solved with exponentially small one-sided error in
\[
   \Oh\bigl((2-\varepsilon_k)^n\bigr)
\]
time. The same bound holds for \(k\)-list-coloring over an arbitrary palette.
\end{theorem}

\begin{proof}
For \(k\leq2\), list-coloring is polynomial-time solvable by 2-SAT~\cite{APT1979}, even over an arbitrary palette: introduce a Boolean variable for the choice at each vertex, and forbid each pair of adjacent choices that use the same color. Singleton lists give forced choices. For \(k\geq3\), apply Theorem~\ref{thm:general-bootstrap} successively for list sizes \(3,4,\ldots,k\). The saving at every level depends only on the maximum list size, so the output algorithm at one level satisfies precisely the hypothesis needed at the next. Taking \(L(v)=[k]\) for every vertex gives ordinary \(k\)-coloring.
\end{proof}

We remark that ordinary $K$-coloring is the special case $L(v)=P=[K]$ for every vertex. Conversely, list coloring over a palette of size $K$ reduces to $K$-coloring by adding a $K$-clique representing the colors and joining each vertex to the clique vertices corresponding to its forbidden colors. Thus, for fixed $K$, the two formulations are equivalent up to $K$ additional vertices; on the other hand, for an unbounded~$K$ the~$k$-list-coloring problem is strictly richer.

\section{Discussion and open problems}\label{sec:discussion}

Our main contribution is a resolution to the long-standing natural question about the exact running times of graph coloring algorithms; we showed that for any~$k\in \mathbb{N}$ there exists~$\varepsilon_k>0$ such that~$k$-coloring can be solved in~$(2-\varepsilon_k)^n$ time. This comes in comparison to the~$\Oh(2^n)$ time algorithm that generally computes the chromatic number of a graph.
Our algorithm also solves $k$-list-coloring for fixed $k$, even when the overall palette is unbounded.

In Appendix~\ref{sec:constants}, we repeat the algorithm while performing a tedious bookkeeping, which shows that without further optimization, the savings~$\varepsilon_k$ given by our algorithms are very small: they behave like the reciprocal of an exponential tower of height~$\Theta(k)$.
The most immediate quantitative problem is thus to improve the asymptotic behavior of $\varepsilon_k$. For~$k$-SAT, which exhibits a similar behavior (general SAT has a~$\Oh(2^n)$ time algorithm while every fixed~$k$-SAT has an algorithm running in~$(2-\varepsilon'_k)^n$ time), an extensive body of works accumulated across decades to improve the constants~$\varepsilon'_k$ or conjecture their asymptotic behavior~\cite{MoSp5,Rodosek96,PPZ99,PPSZ05,schoning1999probabilistic,Hertli14a,Hertli14b,ScSt17,hansen2019faster,ImPa01}. We expect a similar progress to now be possible for~$k$-coloring. The reciprocal-tower estimate in Section~\ref{sec:constants} seems unlikely, or at least currently unjustified, to describe the intrinsic complexity of $k$-coloring.

A separate question remains for binary-list interpolation, where the other lists need not have bounded size. If $b$ of the $n$ vertices have lists of size at most two, what is the smallest palette-independent $q\in (1,2)$ for which one can achieve running time
\[
   \Oh\bigl(2^{n-b}q^b\bigr)?
\]
Theorem~\ref{thm:binary-interpolation} gives $q_K=2^{1-1/\binom K2}$ on a palette of size $K$.  Under ETH, we excluded $q=1$. Is there a palette-independent constant $q<2$? A positive answer would simplify the warm-up reductions and give an interpolation bound independent of both the palette size and the sizes of the other lists.

Two further natural questions are whether the randomization can be removed and whether a sub-\(2^n\) running time can be achieved using only polynomial space. The latter is open even for computing the chromatic number: all known \(\Oh(2^n)\)-time algorithms use exponential space~\cite{WuEtAl2026,GaspersLee2023}.

\section*{Acknowledgments}
Shortly after the first version of this manuscript was posted, the author was contacted by Kevin Pratt, who had concurrently and independently obtained a sub-$2^n$ algorithm for $k$-coloring and posted his manuscript shortly thereafter~\cite{pratt2026breaking2nbarriergraph}.
The two proofs share some mechanisms and both build on the framework and low-degree algorithm of~\cite{Zamir2021}, but their additional main components are substantially different. In particular, Pratt's approach gives better quantitative bounds on the constants~$\varepsilon_k$, while the recursive approach developed here extends naturally to list coloring.

The first version of this manuscript established the list-coloring extension only when the overall color palette has bounded size. The present revision removes this restriction by a small modification of the argument. This extension was developed after useful discussions with Kevin Pratt as well as AI-assisted exploration.

\bibliographystyle{alpha}
\bibliography{refs}

\appendix

\section{An ETH-based lower bound for size-two lists interpolation}
\label{sec:eth}

We prove the following proposition.

\begin{proposition}[Palette-independent interpolation requires \(q>1\)]
\label{prop:no-q-one}
Suppose that there is a constant $q\in[1,2]$, independent of the palette, such that the following holds for every fixed finite palette~$P$: an instance over~$P$ with $u$ vertices whose lists have size at least three and $b$ vertices whose lists have size at most two can be solved in
\[
   \Oh\!\left(2^u q^b\right)
\]
time.  Assuming ETH, there is an absolute constant $q_{\mathrm{ETH}}>1$ such that necessarily $q\geq q_{\mathrm{ETH}}$.  In particular, no such algorithm exists with $q=1$.
\end{proposition}

\begin{proof}
We construct a reduction from bounded-frequency general $(d,2)$-CSP. As cited in Section~\ref{sec:background}, Traxler~\cite{Traxler2008} proved that there is an absolute constant $c>0$ such that, for every fixed domain size $d$, there is a constant $F(d)$ for which $(d,2)$-CSP on $n$ variables requires time $d^{cn}$ under ETH even when every variable occurs in at most $F(d)$ constraints. The important points for us are that $c$ is independent of~$d$ and that $F(d)$ is independent of~$n$.

Fix a domain size~$d$, and consider such a bounded-frequency CSP instance with variable set~$\mathcal X$. We build a list-coloring instance out of it. Absorb every unary constraint into a set $A_x\subseteq[d]$ of allowed values (colors) for its variable~$x$.  If some $A_x$ is empty, the instance is immediately unsatisfiable.  Every binary constraint on distinct variables $x,y$ is equivalently the conjunction of its forbidden assignments
\[
   (x\neq a)\lor(y\neq b),
\]
one for every pair $(a,b)\in[d]^2$ disallowed by that constraint.

Let $J$ be the \emph{primal graph} of the CSP: its vertices are the variables, and two variables are adjacent if they occur together in a binary constraint.  Since every variable occurs in at most $F(d)$ constraints, $J$ has maximum degree at most $F(d)$.  We can therefore greedily compute a proper coloring
\[
   \tau:\mathcal X\longrightarrow[F(d)+1].
\]
The value $\tau(x)$ will serve only as a role identifying the variable $x$ among the variables that may interact with it.

Construct a list-coloring instance over the palette
\[
   P_d=\{\bot\}\cup\bigl([F(d)+1]\times[d]\bigr).
\]
Since $d$ is fixed, this is a fixed palette.  For each CSP variable $x\in\mathcal X$, introduce a \emph{variable vertex} $v_x$ with list
\[
   L(v_x)=\{(\tau(x),a):a\in A_x\}.
\]
For every $x\in\mathcal X$ and every $a\in[d]$, introduce a \emph{selector vertex} $s_{x,a}$, add the edge $(v_x,s_{x,a})$, and set
\[
   L(s_{x,a})=\{\bot,(\tau(x),a)\}.
\]
Finally, for every forbidden assignment $(x=a,y=b)$, add the edge
\[
   (s_{x,a},s_{y,b}).
\]
All selector vertices have lists of size exactly two. See Figure~\ref{fig:eth-selector-reduction} for an illustration of the construction.

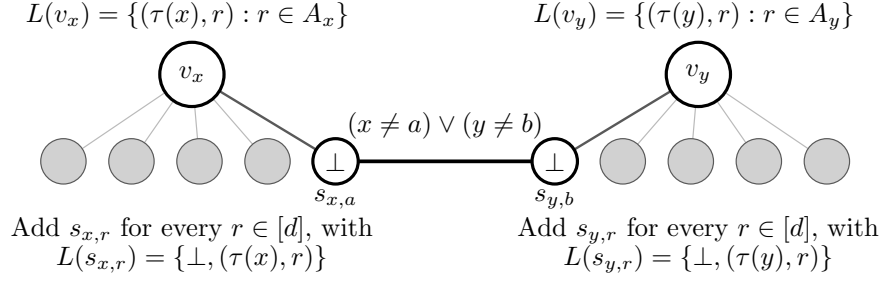
\begin{figure}[t]
\centering
\begin{tikzpicture}[
    x=1cm,y=1cm,
    every node/.style={font=\small},
    varvertex/.style={
      circle,draw=black,very thick,fill=white,
      minimum size=8.5mm,inner sep=0pt
    },
    selector/.style={
      circle,draw=black!65,line width=.45pt,fill=black!18,
      minimum size=6mm,inner sep=0pt
    },
    active/.style={
      circle,draw=black,very thick,fill=white,
      minimum size=6mm,inner sep=0pt
    },
    incidence/.style={draw=black!27,line width=.45pt},
    forcing/.style={draw=black!65,line width=.9pt},
    constraint/.style={draw=black,line width=1.4pt}
]

\coordinate (vx) at (-3.35,.90);
\coordinate (vy) at ( 3.35,.90);

\coordinate (xa) at (-1.45,-.25);
\coordinate (x1) at (-5.05,-.25);
\coordinate (x2) at (-4.15,-.25);
\coordinate (x3) at (-3.25,-.25);
\coordinate (x4) at (-2.35,-.25);

\coordinate (yb) at ( 1.45,-.25);
\coordinate (y1) at ( 5.05,-.25);
\coordinate (y2) at ( 4.15,-.25);
\coordinate (y3) at ( 3.25,-.25);
\coordinate (y4) at ( 2.35,-.25);

\foreach \u in {x1,x2,x3,x4}
  \draw[incidence] (vx)--(\u);
\foreach \u in {y1,y2,y3,y4}
  \draw[incidence] (vy)--(\u);

\draw[forcing] (vx)--(xa);
\draw[forcing] (vy)--(yb);
\draw[constraint] (xa)--(yb);

\node[varvertex] at (vx) {\(v_x\)};
\node[varvertex] at (vy) {\(v_y\)};

\foreach \u in {x1,x2,x3,x4,y1,y2,y3,y4}
  \node[selector] at (\u) {};

\node[active] at (xa) {\(\bot\)};
\node[active] at (yb) {\(\bot\)};

\node at (-1.45,-.75) {\(s_{x,a}\)};
\node at ( 1.45,-.75) {\(s_{y,b}\)};

\node[align=center] at (-3.35,1.68) {
  \(L(v_x)=\{(\tau(x),r):r\in A_x\}\)
};
\node[align=center] at (3.35,1.68) {
  \(L(v_y)=\{(\tau(y),r):r\in A_y\}\)
};

\node[align=center] at (-3.35,-1.35) {
  Add \(s_{x,r}\) for every \(r\in[d]\), with\\[-1pt]
  \(L(s_{x,r})=\{\bot,(\tau(x),r)\}\)
};
\node[align=center] at (3.35,-1.35) {
  Add \(s_{y,r}\) for every \(r\in[d]\), with\\[-1pt]
  \(L(s_{y,r})=\{\bot,(\tau(y),r)\}\)
};

\node at (0,.23) {
  \((x\neq a)\lor(y\neq b)\)
};

\end{tikzpicture}
\caption{Encoding one forbidden assignment of a binary CSP constraint.
Each CSP variable \(z\in\{x,y\}\) is represented by a variable vertex
\(v_z\) and one binary-list selector \(s_{z,r}\) for every \(r\in[d]\).
If \(v_x\) receives \((\tau(x),a)\), its edge to \(s_{x,a}\) forces that
selector to receive \(\bot\), and similarly \(y=b\) forces
\(s_{y,b}\) to receive \(\bot\).  The emphasized selector edge therefore
forbids the two assignments simultaneously.  Since
\(\tau(x)\neq\tau(y)\), the selectors' non-\(\bot\) colors do not create
an unintended conflict.}
\label{fig:eth-selector-reduction}
\end{figure}

We claim that the CSP instance is satisfiable if and only if the constructed list instance is colorable.  Suppose first that $f:\mathcal X\to[d]$ is a satisfying CSP assignment.  Color
\[
   v_x\ \text{by}\ (\tau(x),f(x)),
\]
color $s_{x,f(x)}$ by~$\bot$, and color every other selector $s_{x,a}$ by $(\tau(x),a)$.  Every edge between a variable vertex and one of its selectors is proper.  Now consider an edge $(s_{x,a},s_{y,b})$ corresponding to a forbidden assignment.  Its two endpoints cannot both receive~$\bot$, since that would mean $f(x)=a$ and $f(y)=b$.  If both endpoints receive their non-$\bot$ colors, then those colors are also different because $(x,y)\in E(J)$ and~$\tau$ is a proper coloring, hence $\tau(x)\neq\tau(y)$.  Thus the constructed coloring is proper.

Conversely, consider any list coloring of the constructed instance.  The color of $v_x$ uniquely determines a value $f(x)\in A_x$ through
\[
   c(v_x)=(\tau(x),f(x)).
\]
The edge $(v_x,s_{x,f(x)})$ then forces $c(s_{x,f(x)})=\bot$.  If $f$ selected a forbidden assignment $(x=a,y=b)$, both endpoints of the edge $(s_{x,a},s_{y,b})$ would consequently receive~$\bot$, contradicting properness.  Therefore $f$ satisfies every CSP constraint, proving the claim.

It remains to count the two types of lists.  Let
\[
   t=\bigl|\{x\in\mathcal X:|A_x|\leq2\}\bigr|.
\]
There are $n-t$ variable vertices with lists of size at least three and $t$ variable vertices with lists of size at most two.  In addition, there are exactly $dn$ selector vertices, all with binary lists.  Hence the constructed instance has
\[
   u=n-t,
   \qquad
   b=dn+t.
\]
The assumed interpolation algorithm would solve it, for $q\leq2$, in
\begin{align*}
   \Oh\!\left(2^u q^b\right)
   &=\Oh\!\left(2^{n-t}q^{dn+t}\right)\\
   &=\Oh\!\left((2q^d)^n(q/2)^t\right)\\
   &\leq \Oh\!\left((2q^d)^n\right).
\end{align*}

Choose a sufficiently large fixed integer $d_0$ such that $d_0^c>2$ and $\left(d_0^c/2\right)^{1/d_0}<2$, and define
\[
   q_{\mathrm{ETH}}
   =\left(\frac{d_0^c}{2}\right)^{1/d_0}>1.
\]
If $q<q_{\mathrm{ETH}}$, then
\[
   2q^{d_0}<d_0^c,
\]
and the resulting algorithm would contradict Traxler's lower bound for bounded-frequency $(d_0,2)$-CSP.
\end{proof}

This proposition therefore leaves open a still interesting possibility: is there a universal constant $q<2$, independent of the palette size, for size-two list interpolation?

\section{Plucking up the courage to specify the constants}
\label{sec:constants}
In this section, $K$ denotes the maximum list size, with no bound on the palette size. Throughout the paper, we focused on the existence of a positive \(\varepsilon_K\) for Theorem~\ref{thm:main}; we did not, on the other hand, carry any quantitative estimates between steps of the overall algorithm or reductions. Due to the extensive number of `moving parts' and choices of constants throughout the algorithms, as well as these quantitative bounds being rather grim anyway, this choice was useful for readability. In this section, we repeat and revisit the overall analysis, somewhat informally, to present loose bounds for the asymptotic behavior of  \(\varepsilon_K\) as given by our algorithm without further optimizations. This section is thus unnecessary for any of the statements or proofs in the paper. Our only purpose here is to give the adventurous reader a sense of what \(\varepsilon_K\) is guaranteed by this work, as well as present an explicit baseline for future improvements. Proceed at your own risk.

Our bookkeeping shows that the guaranteed saving \(\varepsilon_K\) is roughly described by the reciprocal of a tower of exponentials of height \(\Theta(K)\).  We make no attempt to optimize the choice of constants. The main loss comes from repeatedly invoking the low-degree algorithm of~\cite{Zamir2021}. That algorithm in turn uses the subset-removal lemma of~\cite[Theorem~1.8]{Zamir2021}. As with many removal- and regularity-type combinatorial arguments, the dependence supplied by that lemma is enormous. At every list-size level, the saving inherited from the preceding level determines a degree threshold; the removal lemma is then applied at that threshold, and the saving proved from it is doubly exponentially smaller. Iterating over the list sizes is what produces the tower behavior.

\subsection{The seed solver}

Fix \(3\leq k\leq K\), and consider one step from \((k-1)\)-list-coloring to \(k\)-list-coloring over arbitrary palettes. Suppose that the algorithm already constructed for \((k-1)\)-list-coloring runs in
\[
   \Oh\!\left(2^{(1-\eta)n}\right),
   \qquad
   a:=2^{1-\eta},
\]
where \(0<\eta\leq1\). The saving in the ordinary exponential base is
\[
   2-a=2\bigl(1-2^{-\eta}\bigr)=\Theta(\eta),
\]
so it suffices to follow the base-two exponent saving~\(\eta\).

We briefly recall the two parameters appearing in the seed solver of Lemma~\ref{lem:seed}. The parameter \(\Delta_1\) is the threshold in the definition of a good active vertex: such a vertex has more than \(\Delta_1\) neighbors receiving one common color from its list in the fixed witness coloring. The parameter \(\beta_0\) is the fraction of active vertices that may be bad while the seed solver is still guaranteed to succeed. Writing
\[
   H:=\log(k/a)=\Theta(\log k),
\]
the analysis of Lemma~\ref{lem:seed} gives the seed solver base-two exponent
\[
   1-\eta+
   \left(
      4\frac{\ln\Delta_1+1}{\Delta_1}+\beta_0
   \right)H.
\]
We therefore set
\[
   \beta_0:=\frac{\eta}{4H}
\]
and choose \(\Delta_1\) to be the least sufficiently large integer satisfying
\[
   4H\frac{\ln\Delta_1+1}{\Delta_1}
   \leq\frac{\eta}{4}.
\]
The seed solver then has exponent at most \(1-\eta/2\), and the parameters have the scales
\[
   \beta_0=\Theta\!\left(\frac{\eta}{\log k}\right),
   \qquad
   \Delta_1
   =
   \Theta\!\left(
      \frac{\log k}{\eta}
      \log\frac{\log k}{\eta}
   \right).
\]
We shall also use the immediate lower bound
\[
   \Delta_1
   \geq\frac{16H}{\eta}
   =
   \Omega\!\left(\frac{\log k}{\eta}\right).
\]

\subsection{One bootstrap step}

\paragraph{Packing stars and partitioning the palette.}
If the seed solver is not guaranteed to succeed, more than \(\beta_0n\) active vertices are bad. As in the proof of Theorem~\ref{thm:general-bootstrap}, set
\[
   \alpha:=\frac{\beta_0}{4},
   \qquad
   \rho:=\frac{\beta_0}{16}.
\]
Here \(\alpha n\) is the threshold for invoking the low-degree algorithm on the unused vertices. Since the packing uses at most \(\delta n\leq\beta_0n/4\) vertices, at least \(\beta_0n/2\) bad active vertices remain eligible whenever this exit does not apply. Consequently, the next star is clean with probability at least \(\beta_0/8=2\rho\).

Choose
\[
   r:=\max\left\{2,
       \left\lceil\frac{2(2k+\log(4/\rho))}{\log(3/2)}\right\rceil
       \right\},
   \qquad
   \Delta:=r(1+k\Delta_1).
\]
The role of \(\Delta\) is to ensure that the \(r\) sampled neighbors avoid the at most \(k\Delta_1\) neighbors whose witness colors lie in the center's list with probability at least \(1/4\). The role of \(r\) is to amortize the cost of selecting stars after the palette partition. With
\[
   \gamma:=\frac{\rho}{2}\,4^{-k}(3/4)^r,
   \qquad
   \kappa:=r-\log(2/\gamma),
\]
our choice gives
\[
   \frac{r\log(3/2)}2\leq\kappa\leq r\log(3/2).
\]
In particular, \(\kappa=\Theta(r)\). Since \(\rho=\Theta(\eta/\log k)\),
\[
   r=O\!\left(k+\log\frac{\log k}{\eta}\right),
   \qquad
   \gamma\geq2^{-O(k)}\left(\frac{\eta}{\log k}\right)^{O(1)}.
\]
Combining this with the estimate for \(\Delta_1\), and using \(k\leq K\), gives
\[
   \alpha=\Theta\!\left(\frac{\eta}{\log k}\right),
   \qquad
   \Omega(1/\eta)\leq\Delta\leq\left(\frac K\eta\right)^{O(1)}.
\]
Only these estimates will matter below.

\paragraph{The saving supplied by the low-degree solver.}
It remains to estimate the saving obtained when at least \(\alpha n\) vertices have degree at most~\(\Delta\). Put
\[
   z_\Delta:=-\ln\bigl(1-2^{-(\Delta+1)}\bigr).
\]
The quantitative bounds in the low-degree algorithm of~\cite[Sections~4.3 and~5]{Zamir2021} give a base-two exponent saving of at least a universal constant times
\[
   g_{\mathrm{LD}}
   :=
   \frac{\alpha z_\Delta}{\Delta(\Delta+1)}
   \exp\!\left(-\frac{\Delta^2\ln k}{z_\Delta}\right).
\]
This follows by substituting
\[
   |S|\geq\frac{\alpha n}{\Delta+1},
   \qquad
   C_{\mathrm{rem}}=\frac{\ln k}{z_\Delta},
   \qquad
   \rho(\Delta,C_{\mathrm{rem}})
   >
   \frac{1}{
      \Delta\exp(1+C_{\mathrm{rem}}\Delta^2)}
\]
into the running-time bound proved there.

Now \(z_\Delta=\Theta(2^{-\Delta})\). Together with
\[
   \alpha
   =
   \Omega\!\left(\frac{\eta}{\log K}\right)
   \quad\text{and}\quad
   \Omega(1/\eta)
   \leq\Delta
   \leq
   (K/\eta)^{O(1)},
\]
this gives
\[
   2^{-\,2^{(K/\eta)^{O(1)}}}
   \leq
   g_{\mathrm{LD}}
   \leq
   2^{-\,2^{\Omega(1/\eta)}}.
\]
This is the doubly exponential loss responsible for the final tower.

\paragraph{Combining the costs of the step.}
Take a sufficiently small universal constant \(c_{\mathrm{LD}}>0\), and set
\[
   g:=c_{\mathrm{LD}}g_{\mathrm{LD}},
   \qquad
   C_{\mathrm{LD}}:=2^{1-g},
   \qquad
   \delta:=\min\left\{\frac{\beta_0}{4},\frac{g}{2\log k}\right\}.
\]
We may assume \(0<g\leq1/2\). The low-degree algorithm has running time \(\Oh(C_{\mathrm{LD}}^n)\), and the cost of enumerating the colors of the used vertices before invoking it has base-two exponent at most
\[
   (1-g)(1-\delta)+\delta\log k
   \leq1-g+\delta\log k
   \leq1-g/2.
\]
The seed solver has exponent at most \(1-\eta/2\). Neither cost is multiplied by the exponential number of star selections: only the terminal two-block calls are repeated that many times.

The packing contains \(m=\delta n/(r+1)+O(1)\) stars, and each selection uses \(\ell=\gamma m/(4r)+O(1)\) of them. After accounting for all selections, the terminal exponent is
\[
   1-\frac{\kappa\gamma\delta}{4r(r+1)}.
\]
Thus, for a sufficiently small universal constant \(c_0>0\), we may choose the new exponent saving to be
\[
   \eta_{\mathrm{new}}
   :=c_0\min\left\{\eta,\ g,
       \frac{\kappa\gamma\delta}{r(r+1)}\right\}.
\]
Here \(\eta_{\mathrm{new}}\) denotes the saving certified by these parameter choices. The number of outer trials is polynomial in \(n\) for fixed parameters and does not change this exponent.

The additional factors \(\gamma\), \(1/r\), and \(1/\log k\) cost at most \(2^{O(K)}(K/\eta)^{O(1)}\) in the reciprocal saving. The low-degree term is still the asymptotically dominant loss. For a suitable explicit choice of the constants hidden above,
\[
   2^{-\,2^{(K/\eta)^{O(1)}}}
   \leq\eta_{\mathrm{new}}
   \leq2^{-\,2^{\Omega(1/\eta)}}.
\]

\subsection{Iterating over all list sizes}

For \(h\geq0\), define
\[
   \tw_0(x):=x,
   \qquad
   \tw_{h+1}(x):=2^{\tw_h(x)}.
\]
Fix a universal constant \(c_0>0\) small enough for the running-time analysis above. At every level, define
\[
   \eta_{k+1}:=c_0\min\left\{\eta_k,\ g_{k+1},
       \frac{\kappa_{k+1}\gamma_{k+1}\delta_{k+1}}
            {r_{k+1}(r_{k+1}+1)}\right\},
\]
where the subscripted quantities are the parameters in the step constructing the \((k+1)\)-list-coloring algorithm. Thus \(\eta_k\) is an exponent saving for the algorithm constructed for \(k\)-list-coloring over arbitrary palettes. We start with \(\eta_2=1\), since \(2\)-list-coloring is solvable in polynomial time, and put
\[
   x_k:=\frac1{\eta_k}.
\]
The estimate for one bootstrap step implies that there are universal constants \(c>0\) and \(D\geq1\) such that
\[
   2^{\,2^{c x_k}}
   \leq
   x_{k+1}
   \leq
   2^{\,2^{(Kx_k)^D}}.
\]

We first derive the lower bound on \(x_K\).  Choose a universal constant \(x_0\) such that
\[
   2^{c x}\geq x
   \qquad\text{for every }x\geq x_0.
\]
By choosing \(c_0\) small enough, we may ensure \(x_3\geq x_0\). At every subsequent level,
\[
   x_{k+1}
   \geq
   2^{\,2^{c x_k}}
   \geq
   2^{x_k}.
\]
Iterating this inequality through the remaining list sizes gives
\[
   x_K\geq\tw_{K-O(1)}(2).
\]

For the reverse direction, fix a sufficiently large universal integer \(A\), and set
\[
   h_k:=AK+3(k-2).
\]
We claim inductively that
\[
   x_k\leq\tw_{h_k}(2).
\]
The claim is immediate for \(k=2\).  Suppose it holds at level \(k\), and write \(y=\tw_{h_k}(2)\).  Since \(h_k\geq AK\), the constant \(A\) may be chosen so that \(y\geq K\) and
\[
   (Kx_k)^D
   \leq
   (Ky)^D
   \leq
   y^{2D}
   \leq
   2^y
   =
   \tw_{h_k+1}(2).
\]
The upper recurrence now gives
\[
   x_{k+1}
   \leq
   2^{\,2^{\tw_{h_k+1}(2)}}
   =
   \tw_{h_k+3}(2)
   =
   \tw_{h_{k+1}}(2).
\]
In particular,
\[
   x_K
   \leq
   \tw_{AK+3(K-2)}(2)
   \leq
   \tw_{(A+3)K}(2).
\]

Set \(\varepsilon_K:=2(1-2^{-\eta_K})\). We finally note that
\[
   \frac{x_K}{2\ln 2}
   \leq
   \frac1{\varepsilon_K}
   \leq
   x_K.
\]

This proves the following.

\begin{proposition}[Quantitative estimate]
There are universal constants \(C_1,C_2>0\) such that, for every sufficiently large \(K\), the saving specified by the parameter choices above satisfies
\[
   \tw_{\lfloor C_1K\rfloor}(2)
   \leq
   \frac1{\varepsilon_K}
   \leq
   \tw_{\lceil C_2K\rceil}(2).
\]
In particular, the reciprocal of the saving obtained above has tower height \(\Theta(K)\).
\end{proposition}

\section{AI Storytime}
The proof strategy for the $k$-coloring result was mine; nonetheless, this project was the first time I found interacting with AI tools useful and productive. This non-mathematical section describes the methods I found practical, together with some cautions and broader discussion.

I used OpenAI's ChatGPT 5.6 Sol model with its memory feature turned on. The model frequently searches previous chats, including seemingly unrelated ones. I even `caught' it using constructions I had described in earlier chats without attribution: a reference to the previous chat appeared in the animated `thinking' glimpse but not in the final answer. Thus, my account of the prompts I supplied might not reveal the full context I inadvertently provided the model with.

Approaching this project, I already had a cohesive proof strategy in mind, building upon prior works. Still, as a sanity check and out of curiosity, I initially withheld it from the AI and revealed it only as needed. I began by providing~\cite{Zamir2021} as the main context and~\cite{BHK2009,ZamirCSP2022,Zamir2022,BeigelEppstein2005} as additional reading, then asked it to generalize the attached paper by one step, from~$6$-coloring to~$7$-coloring. After long thinking, the model came back empty-handed.

Next, I gave it a promising lead, similar to what I would give an early-stage graduate student: an informal, concise but complete description of the high-level strategy at the beginning of Section~\ref{sec:warmup}. I explained what goes wrong in the natural extension of~\cite{Zamir2021}, the remaining obstruction of many size-two lists, and the desired reduction to the clean problem of list-coloring with many such lists. I then asked it to (a) verify the reduction, which I had described only informally, and (b) solve the resulting interpolation problem. The first task was easy: the model repeated the proof details from the paper, generalized them by one color, and reached the same obstruction I had in mind. The second was not: it again thought for a while and drew a blank. Generic encouragement or additional thinking time did not help.

I therefore supplied a more concrete direction. Although the fully general formulation is cleaner, our warm-up application has a bounded color palette, so the \emph{same} size-two list must occur on a constant fraction of the variables. I instructed it to focus on this easier case, at which point the model successfully completed the algorithmic argument. The needed observation was clearly legible from the algorithm it provided: once all short lists contain the same two colors, it suffices to determine which induced subgraphs of the remaining vertices are colorable without those colors.

The bad news, or another cautionary tale, is that instead of citing~\cite{BHKK2007}, which I had not provided as context, and using their all-subgraph coloring result as a black box, the model essentially reproduced (or copied) their entire proof with the new observation mixed into it. Thus, while essentially correct, the algorithm was unnecessarily complicated and displayed a severe lack of attribution. A user less familiar with the area could easily believe the AI had invented the convolution framework and unknowingly repeat published work without citation.

Generally, the model almost felt like simulating an early-stage graduate student. I provided proof frameworks and received a good indication of whether the details could be filled in, or of the precise remaining obstruction. Unlike a student, its occasional mistakes were more adversarially hidden, and its attributions unreliable. On the other hand, what might have emerged from a weekly meeting became a rapid prompt--response cycle. This already feels like a noticeable speed multiplier.

An immediate downside is the apparent loss of what I once considered the best problems for training such early-stage students. If a publicly available machine can now turn a sufficiently strong lead into a solution, without the corresponding understanding and effort, then it may soon become unreasonable to train students on such problems.

I also used AI for the arbitrary-palette extension and the figures. One may look at them and ask whether these outputs are anything to be proud of. I would refer that one to the manually drawn figures in my older papers and ask them to appreciate, at least, the relative improvement.

\end{document}